\pdfoutput=1

\newif\ifdraft\draftfalse   
\newif\ifanon\anonfalse     
\newif\ifcamera\cameratrue  
\newif\ifshepherding\shepherdingfalse 
\newif\ifappendix\appendixfalse 
\newif\iffull\fullfalse     
\newif\iflongrefs\longrefsfalse 
\newif\ifbackref\backreffalse 
\newif\ifsooner\soonerfalse
\newif\iflater\laterfalse
\newif\ifprob\probfalse     
\makeatletter \@input{texdirectives.tex} \makeatother

\ifcamera
  \ifshepherding
  \documentclass[acmsmall,review,nonacm]{acmart}
  \else
  \documentclass[acmsmall,screen]{acmart}
  \fi
\else
  \ifanon
  \documentclass[acmsmall,review,anonymous,nonacm]{acmart}
  \else
  \documentclass[acmsmall,review=false,screen=true,nonacm]{acmart}
  \fi
\fi

\acmPrice{}
\acmJournal{PACMPL}
\acmVolume{3}
\acmNumber{ICFP} 
\acmArticle{104}
\acmYear{2019}
\copyrightyear{2019}
\acmMonth{8}
\acmDOI{10.1145/3341708}
\startPage{1}

\ifanon
\setcopyright{none}             
\else
\ifcamera
\setcopyright{rightsretained}
\else
\setcopyright{rightsretained}   
\fi
\fi

\makeatletter
\def\@copyrightpermission{\ifcamera\\\\\\\fi This work is licensed under a \href{https://creativecommons.org/licenses/by/4.0/}{Creative Commons Attribution 4.0 International License}}
\makeatother

\ifcamera
  \ifshepherding
  \makeatletter
  \def\@authorsaddresses{}
  \makeatother
  \fi
\else
\makeatletter
\def\@authorsaddresses{}
\fancypagestyle{firstpagestyle}{%
  \fancyhf{}%
  \renewcommand{\headrulewidth}{\z@}%
  \renewcommand{\footrulewidth}{\z@}%
    \fancyhead[L]{\ifanon\ACM@linecountL\fi}%
}
\fancypagestyle{standardpagestyle}{%
  \fancyhf{}%
  \renewcommand{\headrulewidth}{\z@}%
  \renewcommand{\footrulewidth}{\z@}%
    \fancyhead[LE]{\ifanon\ACM@linecountL\fi\@headfootfont\thepage}%
    \fancyhead[RO]{\@headfootfont\thepage}%
    \fancyhead[RE]{\@headfootfont\@shortauthors}%
    \fancyhead[LO]{\ifanon\ACM@linecountL\fi\@headfootfont\shorttitle}%
    \fancyfoot[RO,LE]{}%
}
\def\@mkbibcitation{}

\makeatother
\fi

\usepackage{afterpage}
\usepackage{color}
\usepackage{amssymb,amsmath,amsfonts}
\usepackage{xspace}
\usepackage{graphicx}
\usepackage{listings}
\usepackage{mathpartir} 
\usepackage{natbib}
\newcommand\citepos[1]{\citeauthor{#1}'s\ [\citeyear{#1}]}
\let\cite=\citep
\usepackage{amsthm}
\usepackage{stmaryrd}
\usepackage{bbold}
\usepackage{enumitem}
\usepackage{tikz}
\usepackage{tikz-cd}
\usepackage[T1]{fontenc} 
\usepackage[utf8]{inputenc}
\usepackage[all]{xy}
\DeclareUnicodeCharacter{00A0}{~}

\usepackage{combelow}

\ifcamera
\usepackage{flushend} 
\fi
\definecolor{darkblue}{rgb}{0.0,0.0,0.3}

\usepackage{hyperref}
\hypersetup{breaklinks=true}
\usepackage{soul} 

\newcommand\maybecolor[1]{\color{#1}}

\newcommand{\aset}[1]{{\ensuremath{\{#1\}}}}

\let\ls\lstinline

\def\Snospace~{\S{}}

\newcommand\fstar{F$^\star$\xspace}

\newcommand\DM{\textsc{DM}\xspace}
\newcommand\SM{\textsc{SM}\xspace}

\definecolor{dkblue}{rgb}{0,0.1,0.5}
\definecolor{dkgreen}{rgb}{0,0.4,0}
\definecolor{dkred}{rgb}{0.6,0,0}
\definecolor{dkpurple}{rgb}{0.7,0,1.0}
\definecolor{purple}{rgb}{0.9,0,1.0}
\definecolor{olive}{rgb}{0.4, 0.4, 0.0}
\definecolor{teal}{rgb}{0.0,0.4,0.4}
\definecolor{azure}{rgb}{0.0, 0.5, 1.0}
\definecolor{gray}{rgb}{0.5, 0.5, 0.5}
\definecolor{dkgrey}{rgb}{0.2, 0.2, 0.2}
\definecolor{lilac}{rgb}{0.70, 0.04, 0.08}

\newcommand{\comm}[3]{\ifdraft{{\color{#1}[#2: #3]}}\fi}

\newcommand{\ch}[1]{\comm{teal}{CH}{#1}}
\newcommand{\guido}[1]{\comm{blue}{Guido}{#1}}
\newcommand{\km}[1]{\comm{dkred}{Kenji}{#1}}
\newcommand{\da}[1]{\comm{azure}{DA}{#1}}
\newcommand{\et}[1]{\comm{dkblue}{ET}{#1}}

\newcommand{\bob}[1]{\comm{olive}{Bob}{#1}}

\newcommand*{\EG}{e.g.,\xspace}
\newcommand*{\IE}{i.e.,\xspace}
\newcommand*{\ETAL}{et al.\xspace}

\newtheorem{definition}{Definition}
\newtheorem{theorem}{Theorem}

\newcommand*{\ii}[1]{\ensuremath{\mathit{#1}}}

\newcommand{\eqdef}[0]{=}
\newcommand\alt[0]{\;|\;}

\newcommand\Pure{\mathrm{Pure}}
\newcommand\Div{\mathrm{Div}}
\newcommand\St{\mathrm{St}}
\newcommand\StT{\mathrm{StT}}
\newcommand\Exc{\mathrm{Exc}}

\newcommand\ExcT{\mathrm{ExcT}}

\newcommand{\Id}{\mathrm{Id}}

\newcommand{\Cont}{\mathrm{Cont}}

\newcommand\EXC{\mathrm{EXC}}
\newcommand\PURE{\mathrm{PURE}}

\newcommand{\pre}{\mathit{pre}}
\newcommand{\post}{\mathit{post}}
\newcommand{\WP}[0]{\mathrm{wp}}

\newcommand\prop{\mathbb{P}}
\newcommand\Cprop{\Cont_\prop}
\newcommand{\pred}[0]{\mathcal{P}\!\mathrm{red}}
\newcommand{\strpost}[0]{\mathcal{SP}\!\mathrm{ost}}
\newcommand{\prepost}[0]{\mathcal{P}\!\mathrm{re}\mathcal{P}\!\mathrm{ost}}
\newcommand{\relprepost}[0]{\mathcal{R}\!\mathrm{el}\mathcal{P}\!\mathrm{re}\mathcal{P}\!\mathrm{ost}}

\newcommand\M{M}
\newcommand\W{W}
\newcommand\T{\mathcal{T}}

\newcommand{\ret}{\texttt{ret}}
\newcommand\retD{\ret^{\mathcal{D}}}
\newcommand\retM{\ret^{\M}}

\newcommand\retW{\ret^{\W}}

\newcommand{\bind}{\texttt{bind}}
\newcommand\bindD{\bind^{\mathcal{D}}}
\newcommand\bindM{\bind^{\M}}

\newcommand\bindW{\bind^{\W}}
\newcommand{\weaken}{\texttt{weaken}}

\newcommand\hoare[3]{\{\,#1\,\}\;#2\;\{\,#3\,\}}

\newcommand{\algop}{\mathtt{op}}
\newcommand{\geneff}{\mathtt{gen}_{\algop}}

\newcommand\retT{\ret^{T}}

\newcommand{\handlewith}{\mathtt{handle}\text{-}\mathtt{with}}

\newcommand{\handlewithi}[1]{\handlewith^{#1}}

\newcommand{\geneffT}{\geneff^{\T}}
\newcommand{\geneffM}{\geneff^{\M}}
\newcommand{\geneffW}{\geneff^{\W}}
\newcommand{\geneffD}{\geneff^{\mathcal{D}}}

\newcommand{\algopT}{\algop^{T}}
\newcommand{\algopM}{\algop^{\M}}
\newcommand{\algopW}{\algop^{\W}}
\newcommand{\algopD}{\algop^{\mathcal{D}}}

\newcommand\bool[0]{\mathbb{B}}
\newcommand{\true}[0]{\texttt{true}}
\newcommand{\false}[0]{\texttt{false}}
\newcommand{\choice}[0]{\texttt{pick}}
\newcommand{\fail}[0]{\texttt{fail}}
\newcommand{\NDet}[0]{\mathrm{NDet}}
\newcommand{\IO}[0]{\mathrm{IO}}
\newcommand{\Prob}[0]{\mathrm{Prob}}
\newcommand{\supp}[1]{\mathrm{supp}(#1)}

\newcommand{\inr}{\texttt{inr}}
\newcommand{\inl}{\texttt{inl}}

\newcommand{\stget}[0]{\texttt{get}}
\newcommand{\stput}[0]{\texttt{put}}

\newcommand{\stmod}[0]{\texttt{modify}}

\newcommand{\throw}[0]{\texttt{throw}}

\newcommand{\ioread}[0]{\texttt{input}}
\newcommand{\iowrite}[0]{\texttt{output}}

\newcommand{\flip}{\texttt{flip}}

\newcommand{\HistST}{\mathrm{HistST}}
\newcommand{\Hist}{\mathrm{Hist}}
\newcommand{\Fr}{\mathrm{Fr}}

\newcommand\bL{\mathcal{L}}
\newcommand\bM{\mathbb{M}}
\newcommand\vdL{\vdash_{\bL}}
\newcommand\vdD{\vdash_\SM{}}

\newcommand\pair[2]{\left \langle  #1,\;#2 \right \rangle}
\newcommand\proj[1]{\pi_{#1}}
\newcommand\app[2]{#1\;{#2}}
\newcommand\abs[3]{\lambda^{#1} #2.~#3}
\newcommand\substp[2]{{#1}[#2]}
\newcommand\subst[3]{\substp {#1} {#2/#3}}
\newcommand{\tone}[0]{t_{1}}
\newcommand{\ttwo}[0]{t_{2}}
\newcommand{\Type}[0]{\mathrm{Type}}
\newcommand{\one}[0]{\mathbb{1}}
\newcommand{\zero}[0]{\mathbb{0}}

\newcommand{\elab}[2]{\llbracket #2 \rrbracket_{#1}}
\newcommand{\llbrace}[0]{\lbrace\!|}
\newcommand{\rrbrace}[0]{|\!\rbrace}
\newcommand{\elabrelbase}[1]{\llbrace #1 \rrbrace}
\newcommand{\elabrel}[4]{\elabrelbase{#4}^{#3}_{#1,#2}}

\newcommand{\lift}[0]{\texttt{lift}}

\allowdisplaybreaks

\renewcommand{\paragraph}[1]{\smallskip{\bf #1.}\;}

\begin{document}


\setlength{\abovedisplayskip}{2pt}
\setlength{\belowdisplayskip}{2pt}

\def\titlestr{Dijkstra Monads for All}
\title[\titlestr]{\titlestr}




\ifanon
\else
\author{Kenji Maillard} 
\affiliation{
  \ifcamera\institution{Inria}\city{Paris}\country{France}
  \else\institution{Inria Paris}\fi}
\affiliation{
  \ifcamera\institution{ENS}\city{Paris}\country{France}
  \else\institution{ENS Paris}\fi}
\author{Danel Ahman} 
\affiliation{
  \department{Faculty of Mathematics and Physics}
  \ifcamera\institution{University of Ljubljana}\city{Ljubljana}\country{Slovenia}
  \else\institution{University of Ljubljana}\fi}
\author{Robert Atkey} 
\affiliation{
  \ifcamera\institution{University of Strathclyde}\city{Glasgow}\country{UK}
  \else\institution{University of Strathclyde}\fi}
\author{Guido Mart\'inez} 
\affiliation{
  \ifcamera\institution{CIFASIS-CONICET}\city{Rosario}\country{Argentina}
  \else\institution{CIFASIS-CONICET Rosario}\fi}
\author{Cătălin Hriţcu}
\affiliation{
  \ifcamera\institution{Inria}\city{Paris}\country{France}
  \else\institution{Inria Paris}\fi}
\author{Exequiel Rivas} 
\affiliation{
  \ifcamera\institution{Inria}\city{Paris}\country{France}
  \else\institution{Inria Paris}\fi}
\author{\'Eric Tanter} 
\affiliation{
  \department{Computer Science Department (DCC)}
  \ifcamera
  \institution{University of Chile}\city{Santiago}\country{Chile}
  \else\institution{University of Chile}\fi}
\affiliation{
  \ifcamera\institution{Inria}\city{Paris}\country{France}
  \else\institution{Inria Paris}\fi}
\makeatletter
\renewcommand{\@shortauthors}{Maillard~\ETAL}
\makeatother
\fi


\begin{abstract}
This paper proposes a general semantic framework for verifying
programs with arbitrary monadic side-effects using Dijkstra monads,
which we define as monad-like structures indexed by a specification
monad.  We prove that any monad morphism between a computational monad
and a specification monad gives rise to a Dijkstra monad, which
provides great flexibility for obtaining Dijkstra monads tailored to
the verification task at hand.  We moreover show that a large variety
of specification monads can be obtained by applying monad transformers
to various base specification monads, including predicate transformers
and Hoare-style pre- and postconditions.  For defining correct monad
transformers, we propose a language inspired by Moggi's monadic
metalanguage that is parameterized by a dependent type theory.
We also develop a notion of algebraic operations for Dijkstra monads,
and start to investigate two ways of also accommodating effect
handlers. We implement our framework in both Coq and \fstar{}, and
illustrate that it supports a wide variety of verification styles for
effects such as exceptions, nondeterminism, state, input-output, and
general recursion.
\end{abstract}


\ifcamera\ifshepherding\else
\fi\fi

\maketitle

\section{Introduction}
\label{sec:intro}

%


The aim of this paper is to provide a semantic framework for specifying
and verifying programs with arbitrary side-effects modeled by
computational monads~\cite{Moggi89}.
We base this framework on Dijkstra monads, which have already proven valuable
in practice for verifying effectful code~\cite{mumon, evercrypt-msr}.
A Dijkstra monad $\mathcal{D}\;A\;w$ is a monad-like structure that
classifies effectful computations returning values in $A$ and
specified by $w : \W A$, where $W$ is what we call a
\emph{specification monad}.\footnote{Prior work has used the term
  ``Dijkstra monad'' both for the indexed structure $\mathcal{D}$ and for
  the index $\W$~\cite{fstar-pldi13, mumon, dm4free, jacobs14dijkstra, Jacobs15}.
  In order to prevent confusion, we use the term ``Dijkstra monad''
  exclusively for the indexed structure $\mathcal{D}$ and the term
  ``specification monad'' for the index $W$.}
A typical specification monad contains predicate transformers mapping postconditions to preconditions.
For instance, for computations in the state monad $\St\,A = S \to A \times S$,
a natural specification monad is 
$W^{\St}A = (A \times S \to \prop) \to (S \to \prop)$,
mapping postconditions,
which in this case are predicates on final results and states, to
preconditions, which are predicates on initial states
(here $\prop$ stands for the internal type of propositions).
%
%
However, given an {\em arbitrary} monadic effect,
how do we find such a specification monad?
Is there a {\em single} specification monad
that we can associate to each effect?
If not, what are the
{\em various} alternatives, and what 
are the constraints on this {\em association} for obtaining a proper
Dijkstra monad?

A partial answer to this question was provided by the 
{\em Dijkstra Monads for Free} ({\em DM4Free}) approach
of \citet{dm4free}: from a computational monad defined as a term in a
metalanguage called \DM{}, a (single) canonical specification monad is
automatically derived through a syntactic translation. Unfortunately, 
 while this approach works for stateful and
exceptional computations, it cannot handle several other effects, such as input-output (IO), 
due to various syntactic restrictions in \DM{}.

To better understand and overcome such limitations, we make
the novel observation that a computational monad in \DM{} is essentially 
a monad transformer applied to the identity monad; and that the 
specification monad is obtained by applying this monad transformer to the
continuation monad 
${\Cprop A = (A \to \prop) \to \prop}$.
Returning to the example of state, the specification monad $W^{\St}A$
can be obtained from the state monad transformer
$\StT\,M\,A = S \to M (A \times S)$.
This reinterpretation of the {\em DM4Free} approach sheds light on its
limitations: For a start, the class of supported computational monads
is restricted to those that can be decomposed as a monad transformer
applied to the identity monad. However, this rules out various effects such as
nondeterminism or IO, for which no proper monad
transformer is known~\cite{HylandLPP07, AdamekMBL12, BowlerGLS13}.
\ch{So what's the status of this?
  (\href{https://fstar.zulipchat.com/\#narrow/stream/186739-dm4all/topic/Nondeterminism.20transformer}{Zulip discussion})
  It seems to be confusing many people!}

Further, obtaining both the computational and 
specification monads from the same monad transformer introduces
a very tight coupling. In particular, in {\em DM4Free}
one cannot associate {\em different}
specification monads with a particular effect.
For instance, the exception monad $\Exc\,A = A + E$ is associated by {\em DM4Free}
with the specification monad ${W^\Exc A = ((A + E) \to \prop) \to \prop}$, by applying the 
exception monad transformer $\ExcT\,M\,A = M(A + E)$ to $\Cprop$.
This specification monad requires the postcondition to account for 
both the success and failure cases.
While this is often desirable, at times it may be more
convenient to use the simpler specification monad $\Cprop$ directly,
allowing exceptions to be thrown freely, without having to explicitly
allow this in specifications.
Likewise, for IO, one may wish to have rich specifications that 
depend on the {\em
  history} of interactions with the external world, or simpler {\em
  context-free} specifications that are as
local as possible.
In general, one should have the freedom to choose a specification monad that is
expressive enough for the verification task at hand, but also simple
enough so that verification is manageable in practice.

Moreover, even for a fixed computational monad and a fixed
specification monad there can be more than one way to associate the two in a
Dijkstra monad.
For instance, to specify exceptional computations using $\Cprop$, we
could allow all exceptions to be thrown freely---as explained above,
which corresponds to a {\em partial correctness} interpretation---but
a different choice is to prevent any exceptions from being
raised at all---which corresponds to a {\em total correctness} interpretation.
%
Similarly, for specifying nondeterministic computations,
two interpretations are possible for $\Cprop$: a {\em demonic} one, in which the
postcondition should hold for {\em all} possible result values~\cite{Dijkstra75},
and an {\em angelic} one, in which the postcondition should hold for
{\em at least one} possible result~\cite{Floyd67}.

{\em The key idea of this paper is to decouple the computational monad and
the specification monad}: instead of insisting on deriving both from
the same monad transformer as in {\em DM4Free},
we consider them independently and only require that they are related
by a {\em monad morphism}, \IE a mapping between two monads that
respects their monadic structure.
%
%
%
%
For instance, a monad morphism from nondeterministic computations could map a
finite set of possible outcomes 
to a predicate transformer in $(A \to \prop) \to \prop$.
%
Given a finite set $R$ of results in $A$ and a postcondition
$post : A \to \prop$, there are only two reasonable ways to obtain a
single proposition: either take the {\em conjunction} of $post~v$ for every
$v$ in $R$ (demonic nondeterminism), or the {\em disjunction}
(angelic nondeterminism).
For the case of IO, in our framework we can consider at least two
monad morphisms relating the $\IO$ monad to two different
specification monads, $\W^{\Fr}$ and
$\W^{\Hist}$, where $\mathcal{E}$ is the alphabet of IO events:
\[
  \W^{\Fr}X = {(X \times \mathcal{E}^{*} \to \prop) \to \prop} \quad
  \longleftarrow \quad
  \IO
  \quad
  \longrightarrow \quad
  \W^{\Hist}X = {(X \times \mathcal{E}^{*} \to \prop) \to
  (\mathcal{E}^{*} \to \prop)}
\]
While both specification monads take postconditions of the same type
(predicates on the final value and the produced IO events), the produced
precondition of $\W^{\Hist}X$ has an additional argument
$\mathcal{E}^{*}$, which denotes the history of of interactions 
(i.e., IO events) with the external world.

%
This paper makes the following {\bf contributions}:
\begin{itemize}[leftmargin=*,nosep,label=$\blacktriangleright$]
\item We propose a new semantic framework for verifying programs with
  arbitrary monadic effects using Dijkstra monads.
  By decoupling the computational monad from the specification monad
  we remove all previous restrictions on supported computational
  monads. 
  Moreover, this decoupling allows us to flexibly choose the
  specification monad and monad morphism most suitable for the
  verification task at hand.
  We investigate a large variety of specification monads that
  are obtained by applying monad transformers to various base monads,
  including 
  predicate transformers (\EG weakest preconditions
  and strongest postconditions) and Hoare-style pre- and postconditions.
  This flexibility allows a wide
  range of verification styles for nondeterminism, IO, and general recursion---none of which
  was possible with {\em DM4Free}.


\item We give the first general definition of Dijkstra monads as
  a monad-like structure indexed by a specification monad 
  ordered by precision.
  We show that any monad morphism 
  gives rise to a Dijkstra monad, and that from any
  such Dijkstra monad we can recover the monad morphism.
  More generally, we construct an adjunction between
  Dijkstra monads and a generalization of monad morphisms, 
  monadic relations, which induces the above-mentioned equivalence.

\item We recast {\em DM4Free}
  as a special case of our new framework.
  For this, we introduce \SM{}, a principled metalanguage for defining
  correct-by-construction monad transformers.
  The design of \SM{} is inspired by \DM{} and Moggi's monadic metalanguage,
  but it is parameterized by an arbitrary dependent type theory instead of a 
  set of simple types.
  We show that under a natural linearity condition 
  \SM{} terms give rise to correct-by-construction monad transformers
  (satisfying all the usual laws)
  as well as canonical monadic relations, defined from a logical relation.
  This allows us to reap the benefits of the {\em
    DM4Free} construction when it works well (\EG state, exceptions), and
  to explicitly provide monad morphisms when it does not (\EG nondeterminism, IO).

\item We give an account of Plotkin and Power's algebraic operations 
for Dijkstra monads.
We show that a monad morphism equips both its specification monad 
and the corresponding Dijkstra monad with algebraic operations.
We also start to investigate two approaches to effect handlers.
The first approach, in which the specification of operations is induced by the handler, 
allows us to both provide a uniform treatment of {\em DM4Free}'s hand-rolled 
examples of exception handling, and subsume the prior work on 
weakest exceptional preconditions. However, this approach  
seems inherently limited to exceptions.
The second approach, in which operations have to be given specifications
upfront, enables us to also accommodate handlers for effects other than
exceptions, for instance for general recursion, based on McBride's
free monad technique.



\item We illustrate the generality of our semantic framework by
  applying it to the verification of simple monadic programs in both
  Coq and \fstar{}.
\end{itemize}

\paragraph{Paper structure}
We start by reviewing
the use of monads in effectful programming and the closest related approaches
for reasoning about such programs (\autoref{sec:background}).
We then give a gentle overview of our approach through illustrative
examples (\autoref{sec:overview}).
After this, we dive into the technical details:
First, we show how to obtain a wide range of specification
monads by applying monad transformers 
to base specification monads (\autoref{sec:metalanguage}).
Then, we show the tight and natural correspondence between
Dijkstra monads, and monadic relations and monad morphisms (\autoref{sec:dijkstraMonads}).
We also study algebraic operations and effect handlers for Dijkstra
monads (\autoref{sec:algebraic}).
Finally, we outline our implementations of these ideas in \fstar{} and Coq
(\autoref{sec:implementation}), before discussing related
(\autoref{sec:related}) and future work (\autoref{sec:conclusion}).

\ifanon
\else
Supplementary materials include:
(1) verification examples and implementation of our framework in \fstar{}
(\url{https://github.com/FStarLang/FStar/tree/dm4all/examples/dm4all}); 
(2) verification examples and a formalization in Coq
    (\url{https://gitlab.inria.fr/kmaillar/dijkstra-monads-for-all}); 
(3) an \ifappendix\else online \fi appendix with further technical
details\ifappendix\else{} (\url{https://arxiv.org/abs/1903.01237})\fi.
\fi




\section{Background: Monads and Monadic Reasoning}
\label{sec:background}

We start by briefly reviewing the use of monads in effectful
programming, as well as the closest related approaches for verifying
monadic programs.

\subsection{The Monad Jungle Book}
\label{sec:jungle}

Side effects are an important part of programming.
They arise in a multitude of shapes, be it imperative algorithms,
nondeterministic operations, potentially diverging computations, or interactions with the external world.
These various effects can be uniformly captured by the algebraic structure known 
as a {\em computational monad}~\cite{Moggi89, BentonHM00}.
This uniform interface is provided via a type $\M A$ of computations
returning values of type $A$; 
a function ${\retM : A \to \M A}$ that coerces a value $v:A$ to a trivial computation, for
instance seeing $v$ as a stateful computation leaving the state untouched; 
and a function $\bindM\;m\;f$ that sequentially composes the monadic computations
$m:\M A$ with $f : A \to \M B$, for instance threading through the state.
Equations specify that $\retM$ does not have any computational
effect, and that $\bindM$ is associative.

The generic monad interface ($\M,\retM,\bindM$) is, however, not enough to write 
programs that exploit the underlying effect. To this end, each computational 
monad also comes with {\em operations} for causing effects.
%
We briefly recall a few examples of computational monads and their operations:
\begin{description}[nosep]
\item[Exceptions:] 
A computation that can potentially throw exceptions of type $E$
can be represented by the monad $\Exc\,A = A + E$.
Returning a value $v$ is the obvious
left injection, while sequencing $m$ with $f$ is given by 
applying $f$ to $v$ if ${m = \inl\, v}$, or $\inr\, e$ if ${m = \inr\, e}$, i.e., 
when $m$ raised an exception.
The operation
$\throw : E \to \Exc\,\zero$ is defined by right injection. 
%
When we take $E = \one$, exceptions also give us a simple model of partiality 
(the monad being $\Div\,A = A + \one$).


\item[State:] A stateful computation can be
modeled as a state-passing function, \IE $\St\,A = S \to
A \times S$, where $S$ is the type of the state.
Returning a value $v$ is the function $\lambda s . \pair{v}{s}$ that produces
the value $v$ and the unmodified state, whereas binding $m$ to $f$ is obtained
by threading through the state, i.e. $\lambda s . {\bf let}~\pair{v}{s'} =
m~s~{\bf in}~ f~v~s'$. The state monad comes with operations $\stget : \St\,S =
\lambda s . \pair{s}{s}$ to retrieve the state, and $\stput : S \to \St\,\one =
\lambda s . \lambda s' . \pair{\ast}{s}$ to overwrite it.

\item[Nondeterminism:] A nondeterministic computation can be represented by a
  finite set of possible outcomes, i.e. $\NDet\,A = \mathcal{P}_{\mathsf{fin}}(A)$.
Returning a value $v$ is provided by the singleton $\aset{v}$, whereas
sequencing $m$ with $f$ amounts to forming the union $\bigcup_{v \in m}f\,v$.
This monad comes with an operation $\choice : \NDet\,\bool
= \left\{ \true, \false \right\}$, which nondeterministically chooses
a boolean value, and an operation $\fail : \NDet\,\zero = \emptyset$,
which unconditionally fails.
%

\item[Interactive input-output (IO):]
An interactive computation with input type $I$ and output type $O$
can be represented by the inductively defined monad
$\IO\,A = \mu Y . A + (I \to Y) + O \times Y$, which describes
three possible kinds of computations: either return a value ($A$),
expect to receive an input and then continue ($I \to Y$),
or output and continue ($O \times Y$). Returning $v$ is constructing a leaf, whereas sequencing $m$ with $f$
amounts to tree grafting: replacing each leaf with value $a$ in $m$
with the tree $f a$. 
The operations for IO are $\ioread : \IO\,I$ and
$\iowrite : O \to \IO\,\one$.

\end{description}

\subsection{Reasoning About Computational Monads}
\label{sec:predicateTransformers}

Many approaches have been proposed for reasoning about effectful
programs; 
we review the ones closest to ours.
In an imperative setting, Hoare introduced a \emph{program logic} to
reason about properties of programs~\cite{Hoare69}.
The judgments of this logic are \emph{Hoare triples} 
$\hoare{\pre}{c}{\post}$.
Intuitively, if the
precondition $\pre$ is satisfied, then running the program $c$ 
leaves us in a situation where $\post$ is satisfied, provided that $c$
terminates.
%
For imperative programs---\IE statements changing the program's
state---$\pre$ and $\post$ are predicates over
\iffull\else states\fi.

Hoare's approach can be directly adapted to the monadic setting by 
replacing imperative programs $c$ with monadic computations $m : \M\,A$.
This approach was first proposed in Hoare Type Theory~\cite{nmb08htt},
where a {\em Hoare monad} of the form 
$\texttt{HST}\; \pre\; A\; \post$ augments the state monad over $A$
with a precondition $\pre : S \to \prop$ and postcondition
$\post : A \times S \to \prop$.
So while preconditions are still predicates over initial states,
postconditions are now predicates over both final states and results.
%
%
While this approach was successfully extended to a few other
effects~\cite{NanevskiBG13, ynot-icfp08, DelbiancoN13}, there is still
no general story on how to define a Hoare monad or even just the shape
of pre- and postconditions for an arbitrary effect.


A popular alternative to proving properties of imperative programs is
Dijkstra's {\em weakest precondition calculus}~\cite{Dijkstra75}. The
main insight of this calculus is that we can typically compute a
weakest precondition $\WP(c, post)$ such that $\pre \Rightarrow \WP(c,
post)$ if and only if $\hoare{\pre}{c}{\post}$, and therefore partly
automate the verification process by reducing it to a logical
decision problem. \citet{fstar-pldi13} observed that it is
possible to adopt Dijkstra's technique to ML programs with state
and exceptions elaborated to monadic style. They
propose a notion of {\em Dijkstra monad} of the form
$\texttt{DST}\; A\; wp$, where $wp$ is a predicate transformer that
specifies the behavior of the monadic computation. These predicate
transformers are represented as functions that, given a postcondition
on the final state, and the result value of type $A$ or an exception of type $E$,
calculate a corresponding precondition on the initial state. Their
predicate transformer type can be written as follows:
\[
  W^{\ii{ML}} \quad=\quad \underbrace{((A + E)  \times S \to \prop)}_{\mathit{postconditions}}\; \to \underbrace{(S \to \prop)}_{\mathit{preconditions}}.
\]
In subsequent work, \citet{mumon} extend this to programs that
combine multiple sub-effects. They compute more efficient weakest
preconditions by using Dijkstra monads that precisely capture the actual
effects of the code, instead of verifying everything using $W^{\ii{ML}}$ above.
For example, pure computations are verified using a Dijkstra monad
whose specifications have type:
\[
  W^{\Pure}A \quad=\quad\Cprop\,A\quad=\quad (A \to \prop) ~\to~ \prop,  \\
\]
while stateful (but exception-free) computations are verified using
specifications of type:
\[
  W^{\St}A \quad=\quad (A \times S \to \prop)  ~\to~ (S \to \prop).
\]
Recently, \citepos{dm4free} {\em DM4Free} work shows that these
originally disparate specification monads can be uniformly derived
from computational monads defined in their \DM{} metalanguage.

An important observation underlying these techniques is that predicate
transformers have a natural monadic structure. For instance, it is not
hard to see that the predicate transformer type $\W^{\Pure}$ is
simply the continuation monad with answer type $\prop$, that
$\W^{\St}$ is the state monad transformer applied to $\W^{\Pure}$, and
that $\W^{\ii{ML}}$ is the state and exceptions monad transformers
applied to $\W^{\Pure}$.
It is this monadic structure that supports writing computations that carry their own specification. In the next section, we show that it is also the basis for what we call a \emph{specification monad}.

\section{A Gentle Introduction to Dijkstra Monads for All}
\label{sec:overview}

\ch{To save space the amount of signposting could be reduced in this
  section. It's anyway a bit excessive: ``let us spend 3 lines to
  explain what we're going to tell you in the next 5 line paragraph''
  kind of stuff}

In this section we introduce a few basic definitions and illustrate
the main ideas of our semantic framework on various relatively simple examples.
We start from the observation that the kinds of specifications most
commonly used in practice form {\em ordered monads}
(\autoref{sec:specificationmonads}).
On top of this we define {\em effect observations}, as just monad
morphisms between a computation and a specification monad
(\autoref{sec:observations}), and give various examples
(\autoref{sec:EffObsExamples}).
Finally, we explain how to use effect observations to obtain Dijkstra
monads, and how to use Dijkstra monads for program verification
(\autoref{sec:DijkstraMonadsExamples}).

\subsection{Specification Monads}
\label{sec:specificationmonads}
\label{sec:ordered-monads}

The realization that predicate transformers form
monads~\cite{fstar-pldi13, mumon, dm4free, jacobs14dijkstra, Jacobs15}
is the starting point to provide a uniform notion of specifications.
Generalizing over prior work, we show that this is true not only for
weakest precondition transformers, but also for strongest
postconditions, and pairs of pre- and postconditions
(see \autoref{sec:basic-spec-monads}).
Intuitively, elements of a specification monad can be used to specify
properties of some computation, \EG $\W^{\Pure}$ can specify pure or
nondeterministic computations, and $\W^{\St}$ can specify stateful
computations.

The specification monads we consider are {\em ordered}.
%
%
%
Formally, a monad $\W$ is ordered when $\W A$ is equipped with a
preorder $\leq^{\W A}$ for each type $A$, and $\bindW$ is monotonic in
both arguments:
\[\forall (w_{1} \leq^{\W A} w'_{1}).\, \forall (w_{2}\,w'_{2} : A \to \W B).\, (\forall
  x:A.\, w_{2}\,x \leq^{\W B} w'_{2}\,x) \Rightarrow \bindW~w_{1}~w_{2}\leq^{\W B}\bindW~w'_{1}~w'_{2}\]
This order allows specifications to be compared as being either more or less precise.
For example, for the specification monads $\W^{\Pure}$ and $\W^{\St}$, the ordering is given by
\begin{align*}
w_{1} \leq w_{2} : \W^{\Pure}A \qquad&\Leftrightarrow \qquad \forall
  (p : A \to \prop).\, w_{2}\,p \Rightarrow w_{1}\,p\\
w_{1} \leq w_{2} : \W^{\St}A \qquad&\Leftrightarrow \qquad \forall
  (p : A \times S \to \prop) (s:S).\, w_{2}\,p\,s \Rightarrow w_{1}\,p\,s
\end{align*}
%
%
For $\W^{\Pure}$ and $\W^{\St}$ to form ordered monads,
it turns out that we need to restrict our attention to
\emph{monotonic} predicate transformers, \IE{} those mapping (pointwise) stronger
postconditions to stronger preconditions.
%
%
This technical condition, quite natural from the point of view of
verification, will be assumed implicitly for all the predicate
transformers, and will be studied in detail
in~\autoref{sec:basic-spec-monads}.

%
As explained in \autoref{sec:predicateTransformers},
%
a 
powerful way to construct specification monads is to
apply monad transformers to existing specification monads, e.g.,
applying $\ExcT\,\M\,A = \M\,(A + E)$ to
$\W^{\Pure}$ we get
\[\W^{\Exc}A \quad = \quad \ExcT\,\W^{\Pure}\,A \quad = \quad ((A + E) \to \prop) \to \prop \quad \cong\quad (A \to \prop) \to (E \to \prop) \to \prop \]
$\W^{\Exc}$ is a natural specification monad for programs that can throw exceptions,
transporting a normal postcondition in $A \to \prop$ and an exceptional
postcondition in $E \to \prop$ to a precondition in $\prop$.
%
%
%
Further specification monads using this idea will be introduced along
with the examples in~\autoref{sec:EffObsExamples}.
%

\subsection{Effect Observations}
\label{sec:observations}

Now that we have a presentation of specifications as elements of a
monad, we can relate computational monads to such specifications.
Since an object relating computations to specifications provides a
 particular insight to the potential effects of the computation, they
have been called {\em effect observations}~\cite{Katsumata14}.
As explained in \autoref{sec:intro}, a computational monad can have effect
observations into multiple specification monads, or multiple 
effect observations into a single specification monad.
%
%
Using the exceptions computational monad $\Exc$ as running example,
we argue that {\em monad morphisms} provide a natural notion of
effect observation in our setting, and we provide example monad
morphisms supporting this claim.
%
%
%
Further examples are explored in \autoref{sec:EffObsExamples}.

\paragraph{Effect observations are monad morphisms}
As explained in \autoref{sec:jungle}, computations throwing
exceptions can be modeled by monadic expressions $m : \Exc\,A = A + E$.
A natural way to specify $m$ is to consider the specification monad 
$\W^{\Exc}A = ((A + E) \to \prop) \to \prop$ and to map $m$ to the predicate
transformer $\theta^{\Exc}(m) \eqdef \abs{}{p}{p\,m} : \W^{\Exc}A$, 
applying the postcondition $p$ to the computation $m$.
The mapping $\theta^{\Exc} : \Exc \to \W^{\Exc}$ relating the computational
monad $\Exc$ and the specification monad $\W^{\Exc}$ is parametric in the return
type $A$, and it verifies two important properties with respect to the monadic
structures of $\Exc$ and $W^{\Exc}$.
First, a returned value is specified by itself:
\[ \theta^{\Exc}(\ret^{\Exc}\,v) = \theta^{\Exc}(\inl\,v) =
  \abs{}{p}{p\,(\inl\,v)} = \ret^{\W^{\Exc}}\,v\]
and second, $\theta$ preserves the sequencing
of computations:
\begin{align*}
\theta^{\Exc}(\bind^{\Exc}~(\inl\,v)~f) &= \theta^{\Exc}(f v) = \bind^{\W^{\Exc}}(\ret^{\W^{\Exc}}v)~(\theta^{\Exc}{\circ} f) = \bind^{\W^{\Exc}}~\theta^{\Exc}(\inl\, v)~(\theta^{\Exc} {\circ} f)\\
\theta^{\Exc}(\bind^{\Exc}~(\inr\,e)~f) &= \theta^{\Exc}(\inr\,e) = \bind^{\W^{\Exc}}~\theta^{\Exc}(\inr\,e)~(\theta^{\Exc}\circ f)
\end{align*}
These properties together prove that $\theta^{\Exc}$ is a monad
morphism. More importantly, they allow us to compute
specifications from computations \emph{compositionally}, \EG the
specification of $\bind$ can be computed from the specifications of
its arguments.
This leads us to the following definition:

\begin{definition}[Effect observation]
 An \emph{effect observation} $\theta$ is a monad morphism from a computational
monad $\M$ to a specification monad $\W$. More explicitly, it is a
family of maps ${\theta_A : \M\,A \longrightarrow \W\,A}$, natural in $A$ and
such that for
any $v : A$, $m:\M\,A$ and $f : A \to \M\,B$ the following equations hold:
\begin{align*} \theta_{A}(\ret^{\M}\,v) &= \ret^{\W}\,v
&\theta_{B}(\bind^{\M}\,m\,f)
= \bind^{\W}\,(\theta_{A}\,m)\,(\theta_{B} \circ f) \end{align*}
\end{definition}


\paragraph{Specification monads are not canonical}
When writing programs using the exception monad, we may want to write
pure sub-programs that actually do not raise exceptions.
In order to make sure that these sub-programs are pure, we could use
the previous specification monad and restrict ourselves to
postconditions that map exceptions to false ($\bot$): hence raising an
exception would have an unsatisfiable precondition. However, as
outlined in \autoref{sec:intro}, a simpler solution is possible.
Taking as specification monad $\W^{\Pure}$, we can define the following effect
observation $\theta^{\bot} : \Exc \to \W^{\Pure}$ by
\begin{align*}
  \theta^{\bot}(\inl\,v)&= \abs{}{p}{p\,v}
  & \theta^{\bot}(\inr\,e) &= \abs{}{p}{\bot}
\end{align*}
%
This effect observation{} gives a {\em total correctness} interpretation to
exceptions, which prevents them from being raised at all.
As such, we have effect observations from $\Exc$ to both $\W^{\Exc}$ and $\W^{\Pure}$.

\paragraph{Effect observations are not canonical}
Looking closely at the effect observation $\theta^{\bot}$, it is clear that we
made a rather arbitrary choice when mapping every exception $\inr\,e$ to $\bot$.
Mapping $\inr\,e$ to true ($\top$) instead also gives us an effect observation, 
$\theta^{\top}: \Exc \to \W^{\Pure}$.
This effect observation assigns a trivial precondition to the $\throw$ operation,
providing a \emph{partial correctness} interpretation: given a program $m :
\Exc\,A$ and a postcondition $p:A \to \prop$, if $\theta^{\top}(m)(p)$ is
satisfiable and $m$ evaluates to $\inl\,v$ then $p\,v$ holds; but $m$ may
also raise any exception instead.
Thus, $\theta^{\bot}, \theta^{\top} :
\Exc \to \W^{\Pure}$ are two natural effect observations 
 into the {\em same} specification monad.
Even more generally, we can vary the choice for each exception; in fact,
effect observations $\theta: \Exc \to \W^{\Pure}$ are in one-to-one
correspondence with maps
$E \to \prop$~(see \autoref{sec:monad-algebras} for a general account of this correspondence).

\subsection{Examples of Effect Observations}
\label{sec:EffObsExamples}


When specifying and verifying monadic programs, there is generally a
large variety of options regarding both the specification monads and
the effect observations.
We will now revisit more computational monads from
\autoref{sec:jungle}, and present various natural effect observations for them.
%

%
\paragraph{Monad transformers}
Even though there is, in general, no canonical effect observation for a
computational monad, for the case of a monad $\mathcal{T}(\Id)$ (i.e.,
a monad obtained by the application of a monad transformer to the
identity monad) we can build a canonical specification monad, namely
$\mathcal{T}(\W^\Pure)$, and a canonical effect observation into it.
The effect observation is obtained simply by lifting the $\ret^{\W^{\Pure}} :
\Id \to \W^{\Pure}$ function through the $\mathcal{T}$ transformer.
%
This is the main idea behind our reinterpretation of the
{\em DM4Free} approach~\cite{dm4free}.
For instance, for the exception monad $\Exc = \ExcT(\Id)$ and the
specification monad $\W^{\Exc} = \ExcT\,(\W^{\Pure})$,
%
%
the effect observation
$\theta^{\Exc}$ arises as simply $\theta^{\Exc} = \ExcT(\ret^{\W^{\Pure}})
= \abs{}{m\,p}{p\,m}$.
%
%
%
%
More generally, for any monad transformer
$\mathcal{T}$ (e.g. $\StT, \ExcT, \StT \circ \ExcT, \ExcT \circ \StT$)
and any specification monad $W$ (so not just $\W^\Pure$, but also \EG
any basic specification monad from \autoref{sec:basic-spec-monads})
we have a monad morphism
\[\theta^{\mathcal{T}} \quad:\quad \mathcal{T}(\Id) \quad
  \xrightarrow{\mathcal{T}(\ret^{\W^{\Pure}})} \quad\mathcal{T}(\W^{\Pure})\]
providing effect observations for stateful computations with exceptions, or for
computations with rollback state.
%
%
%
%
However, not all computational monads arise as a monad transformer
applied to the identity monad. The following examples illustrate the
possibilities in such cases. 

\paragraph{Nondeterminism} The computational monad $\NDet$ admits effect
observations to the specification monad $\W^{\Pure}$.
Given a nondeterministic computation $m : \NDet\,A$
represented as a finite set of possible outcomes, and a
postcondition $post : A \to \prop$, we obtain a set $P$ of
propositions by applying $post$ to each element of $m$.
There are then two natural ways to interpret $P$ as 
a single proposition:
\begin{itemize}[leftmargin=*,nosep,label=$\vartriangleright$]
\item we can take the conjunction $\bigwedge_{p\in
      P} p$, which corresponds to the weakest
  precondition such that \emph{any} outcome of $m$ satisfies $\post$
  (\emph{demonic nondeterminism}); or 
\item we can take the disjunction $\bigvee_{\!p \in P} p$,
  which corresponds to the weakest precondition such that
  {\em at least one} outcome of m satisfies $\post$
  (\emph{angelic nondeterminism}).
\end{itemize}
To see that both these choices lead to monad morphisms $\theta^{\forall},
\theta^{\exists} : \NDet \to \W^{\Pure}$, 
it is enough to check that taking the
conjunction when $P = \aset{p}$ is a singleton is
equivalent to $p$, and that a conjunction of conjunctions $\bigwedge_{a\in
  A}\bigwedge_{p \in P_{\!a}} p$ is equivalent to a conjunction on
the union of the ranges $\bigwedge_{p \in \bigcup_{a\in A}P_{\!a}}p$---and similarly for disjunctions. Both conditions are straightforward to check.
%

\paragraph{Interactive Input-Output}
Let us now consider programs in the $\IO$ monad (\autoref{sec:jungle}).
%
We want to define an effect observation $\theta : \IO \to \W$, for some
specification monad $\W$ to be determined.
A first thing to note is that since no equations constrain
the $\ioread$ and $\iowrite$ operations, we can specify their
interpretations $\theta(\ioread) : \W\,I$ and
$\forall (o:O).\, \theta(\iowrite\, o) : \W\,\one$
separately from each other.

Simple effect observations for $\IO$ can already be provided using the
specification monad $\W^{\Pure}$.
%
The interpretation of the $\iowrite$ operation in this simple case needs to
provide a result in $\prop$ from an output element $o:O$ and a postcondition
$p : \one \to \prop$. 
%
Besides returning a constant proposition (like for
$\theta^{\bot}, \theta^{\top}$ in~\autoref{sec:observations}),
a reasonable interpretation is to forget
  the $\iowrite$
operation and return $p\,\ast$ (where $\ast$ is the unit value).
%
%
For the definition of $\theta(\ioread) : (I \to \prop) \to \prop$,
we are given a postcondition $post : I \to \prop$ on the possible inputs 
and we need to build a proposition.
Two canonical solutions are to use either the universal quantification
$\forall (i:I).\, post\,i$, requiring that the postcondition is valid for
the continuation of the program for
any possible input; or the existential quantification $\exists (i:I).\, post\,i$,
meaning that there exists some input such that the
program's continuation satisfies the postcondition, 
analogously to the two modalities 
of evaluation logic~\cite{pitts1991evaluation, Moggi95}.

%
To get more interesting effect observations accounting for
inputs and outputs we can, for instance, extend $\W^{\Pure}$ with {\em
  ghost state}~\cite{OwickiG76} capturing the list of
executed IO events.\footnote{Importantly, the ghost state only appears in 
specifications and not in user programs; these still use only (stateless) $\IO$.}
We can do this by applying the state monad transformer with state type
$\mathrm{list}\,\mathcal{E}$ to $\W^{\Pure}$, obtaining the specification monad
$\W^{\HistST}\,A = (A \times \mathrm{list}\,\mathcal{E} \to \prop) \to \mathrm{list}\,\mathcal{E}
\to \prop$, 
%
for which we can provide interpretations of $\ioread$ and $\iowrite$
that also keep track of the
history of events via ghost state:
\newcommand{\EE}{\mathcal{E}}
\newcommand{\In}{\mathrm{In}}
\newcommand{\Out}{\mathrm{Out}}
\newcommand{\IOFr}{\IO^{\Fr}}
\newcommand{\IOFR}[2]{\mathrm{IOFree}~{#1}~{#2}}
\newcommand{\WFr}{\W^{\Fr}}
\begin{align*}
    \theta^{\HistST}(\iowrite\,o)\quad &= \quad \abs{}{(p:\one \times \mathrm{list}\,\EE
    {\to} \prop)\,(log : \mathrm{list}\,\EE)}{p\, \langle \ast, (\Out\,o) :: log\rangle}\quad &: \W^{\HistST}(\one)\\
    \theta^{\HistST}(\ioread)\quad &= \quad \abs{}{(p:I \times \mathrm{list}\,\EE
    {\to} \prop)\,(log : \mathrm{list}\,\EE)}{\forall i.\, p\, \langle i, (\In\,i) :: log\rangle}\quad &: \W^{\HistST}(I)
\end{align*}

%
%
This specification monad is however somewhat inconvenient in that
postconditions are written over the {\em global} history of events, instead of
over the events of the expression in question.
%
Further, one can write specifications that ``shrink'' the global history
of events, such as $\abs{}{p\,log}{p\,\langle \ast, []\rangle}$, which
\emph{no} expression satisfies.
For these reasons, we introduce an \emph{update monad}~\cite{ahman13update} 
variant of $\W^{\HistST}$, written $\W^{\Hist}$, 
which provides a more concise
way to describe the events. In particular, in $\W^{\Hist}$ the postcondition specifies  
only the events produced by the expression, while the precondition is still free to 
specify any previously-produced events, allowing us to define:
\begin{align*}
    \theta^{\Hist}(\iowrite\,o)\quad &= \quad \abs{}{(p:\one \times \mathrm{list}\,\EE
    {\to} \prop)\,(log : \mathrm{list}\,\EE)}{p\, \langle \ast, [\Out\,o]\rangle}\quad &: \W^{\Hist}(\one)\\
    \theta^{\Hist}(\ioread)\quad &= \quad \abs{}{(p:I \times \mathrm{list}\,\EE
    {\to} \prop)\,(log : \mathrm{list}\,\EE)}{\forall i.\, p\, \langle i, [\In\,i]\rangle}\quad &: \W^{\Hist}(I)
\end{align*}
While $\W^{\Hist} = \W^{\HistST}$, the two monads 
differ in their $\ret$ and $\bind$ functions. For instance, 
\begin{align*}
   \bind^{\W^{\HistST}}\, w\, f \quad & = \quad \abs{} {p~log} {w\, \big(\abs{}{\pair{x}{log'}}{f\, x\, p\, log'}\big)\, log}
   \\
   \bind^{\W^{\Hist}}\, w\, f \quad & = \quad \abs{} {p~log} {w\, \big(\abs{}{\pair{x}{log'}}{f\, x\, \big(\abs{}{(y,log'')}{p\, \langle y , log' +\!\!+ log'' \rangle}\big)\, (log +\!\!+ log')}\big)\, log}
\end{align*}
where the former overwrites the history, while the 
latter merely augments it with new events.

While $\W^{\Hist}$ provides a good way to reason about IO,
some $\IO$ programs do not depend on past interactions.
For these, we can provide an even more parsimonious 
specification monad by applying the writer transformer to $\W^{\Pure}$.
The resulting specification monad $\W^{\Fr}$ then allows us to define
\begin{align*}
  \label{eq:IOFrEffObs}
    \theta^{\Fr}(\iowrite\,o) \quad& =\quad \abs{}{(p:\one \times \mathrm{list}\,\EE
  {\to} \prop)}{p\,\langle  \ast, [\Out\,o] \rangle} \quad:\quad \W^{\Fr}(\one)  \\
    \theta^{\Fr}(\ioread) \quad &=\quad \abs{}{(p:I \times \mathrm{list}\,\EE
    {\to} \prop)}{\forall i.\, p\,\langle  i, [\In\,i] \rangle} \quad:\quad \W^{\Fr}(I) 
\end{align*}
This is in fact a special case of $\W^{\Hist}$ where the history is taken to be
$\one$~\cite{ahman13update}.

\newcommand{\IOSt}{\IO\St}
\newcommand{\WIOSt}{\W^{\IO\St}}

In fact, there is even more variety possible here, 
\EG it is straightforward to write specifications that speak only
of output events and not input events, and vice versa. 
It is also easy to extend this style of reasoning to combinations of 
IO and other effects. For instance, we can 
simultaneously reason about state changes and IO events 
by considering computations in 
$\IOSt\, A = S \to \IO(A \times S)$, resulting from applying 
the state monad transformer to $\IO$, together with the 
specification monad 
$\WIOSt\,A = (A \times S \times \mathrm{list}\,\mathcal{E} \to \prop) \to S \to \mathrm{list}\,\mathcal{E}
\to \prop$.
As such, we recover the style proposed by 
\citet{MalechaMW11}, though they also cover separation logic, which 
we leave as future work.
Being able to choose between specification monads and effect observations
allows one to keep the complexity of the specifications low when the
properties are simple, yet increase it if required.



\subsection{Recovering Dijkstra Monads}
\label{sec:DijkstraMonadsExamples}

\newcommand{\ST}[0]{\mathrm{ST}}
\newcommand{\STD}[2]{\ST~{#1}~{#2}}




We now return to Dijkstra monads, which provide a practical and
automatable verification technique in dependent type theories like
\fstar{}~\cite{mumon}, where they are a primitive notion, and
Coq, where they can be embedded via dependent types.
%
%
We explain how a Dijkstra monad can be obtained from a computational
monad, a specification monad, and an effect observation relating them.
Then we show how the obtained Dijkstra monad can be used for actual
verification.

\paragraph{Stateful computations}
Let us start with stateful computations as an illustrative example, taking
the computational monad $\St$, the specification monad
$\W^\St$, and the following effect observation:
\[
    \begin{array}{rcl}
        \theta^{\St}    & : & \St \rightarrow \W^\St \\
        \theta^{\St}(m) & = & \abs{}{post\,s_{0}}{post\,(m\,s_{0})} \\
    \end{array}
\]
We begin by defining the Dijkstra monad type constructor,
$\ST : (A:\Type) \rightarrow \W^\St~A \rightarrow \Type$.
The type $\STD{A}{w}$ contains all those computations
$c : \St~A$ that are correctly specified by $w$.
%
We say that $c$ is {\em correctly specified} by $w$ when $\theta^{\St}(c) \le w$,
that is, when $w$ is weaker than (or equal to) the specification
given from the effect observation.
%
Unfolding the definitions of $\le$ and $\theta^{\St}$, this
intuitively says that for any initial state $s_{0}$ and postcondition
$post : A \times S \to \prop$, the precondition $w\,post\,s_{0}$
computed by $w$ is enough to ensure that $c$ returns a value $v:A$ and
a final state $s_{1}$ satisfying $post\pair{v}{s_{1}}$; in other
words, $w\,post\,s_{0}$ implies the weakest precondition of $c$.
%
%


The concrete definition for the type of a Dijkstra monad can vary
according to the type theory in question.
For instance, in our Coq development, we define it (roughly) as a
dependent pair of a computation $c : \St~A$
and a proof that $c$ is correctly
specified by $w$.
In \fstar, it is instead a primitive notion.
In the rest of this section, we shall not delve into such representation details.
%


%
%
%
%

The Dijkstra monad $\ST$ is equipped with monad-like\et{why
  "monad-like"? \ch{because monad indexed monads are not plain monads}}
functions $\ret^{\ST}$ and $\bind^{\ST}$ whose definitions come from the
computational monad $\St$, while their specifications come from the
specification monad $\W^{\St}$.
%
%
The general shape for the $\ret$ and $\bind$ of the obtained Dijkstra
monad is:\footnote{If the representation of the Dijkstra monad
is dependent pairs, then the code here does not typecheck as-is and
requires some tweaking. For this section we will assume Dijkstra monads are
defined as \emph{refinements} of the computational monad, without
any explicit proof terms to carry around. In our Coq implementation
we use {\tt Program} and evars to hide such details.}
\[
    \begin{array}{rclcl}
        \ret^{\ST}  & = & \ret^\St & : & (v:A) \rightarrow \STD{A}{(\ret^{\W^{\St}}~v)} \\

        \bind^{\ST} & = & \bind^\St & : & (c:\STD{A}{w_c})
                            \rightarrow (f : (x:A) \rightarrow \STD{B}{(w_f\, x)}) 
                            \rightarrow \STD{B}{(\bind^{\W^\St}\,w_c\,w_f)}
    \end{array}
\]
which, after unfolding the state-specific definitions becomes:
%
\[
  \begin{array}{rclcl}
    \ret^{\ST}  & = & \ret^\St
    & :
    & (v:A) \rightarrow \STD{A}{(\abs{}{post\,s_0}{post~\pair{v}{s_{0}}})} \\
    \bind^{\ST} & = & \bind^\St
    & :
    & (c:\STD{A}{w_c}) \rightarrow (f : (x:A) \rightarrow \STD{B}{(w_f\, x)}) \\
                & & &
    & \rightarrow \STD{B}{(\abs{}{p\,s_0}{w_c~(\abs{}{\pair{x}{s_{1}}}{w_f~x~p~s_1})~s_0})} \\
  \end{array}
\]
%
%

The operations
of the computational
monad are also reflected into the Dijkstra monad, with 
their specifications are computed by $\theta^{\St}$. 
Given $\ii{op}^\St : (x_1:A_1) \rightarrow \cdots \rightarrow (x_n:A_n) \rightarrow \St~B$, 
we can define
%
\[
    \begin{array}{rclcl}
        \ii{op}^{\ST}  & = & \ii{op}^\St & : & (x_1:A_1) \rightarrow \cdots \rightarrow (x_n:A_n) \rightarrow \STD{B}{(\theta^{\St}(\ii{op}^{\St}~x_1~\ldots~x_n))} \\
    \end{array}
\]
%
Concretely, for state, we get the following two operations for the Dijkstra monad $\ST$:
\begin{align*}
  \stget \quad &: \quad  \STD{S}{(\abs{}{p\,s_{0}}{p\,\pair{s_{0}}{s_{0}}})},
  & \stput{}  \quad &: \quad (s:S) \to \STD{\one}{(\abs{}{p\,s_{0}}{p\,\pair{\ast}{s}})}.
\end{align*}
%

%

Given this refined version of the state monad, computing
specifications of (non-recursive) programs becomes simply a matter of
doing {\em type inference} to compositionally lift the program to a
specification and then unfolding the specification by (type-level) computation.
For instance, given
$
    \stmod\,(f : S \rightarrow S) = \bind^{\ST}\,\stget\,(\abs{}{x}{\stput (f x)}),
$
both \fstar{} and Coq can infer the type
\begin{align*}
  \STD{\one}{(\bind^{W^\St}\,(\abs{}{p\,s_{0}}{p\,\pair{s_{0}}{s_{0}}})~
  (\abs{}{s\,p\,s_{0}}{p\,\pair{\ast}{f s}}))} \quad
  &=\quad  \STD{\one}{(\abs{}{p\,s_0}{p\,{\pair{\ast}{f\, s_0}}})}
\end{align*}
which precisely describes the behavior of $\stmod$.
Program verification then amounts to proving 
that, given a programmer-provided type-annotation $\STD \one w$ for 
$\stmod f$, the specification $w$ is weaker than the inferred specification.


\km{From discussions with Nikita Zuslin, it seems that he had difficulties to
  understand how Dijkstra monads actually do the verification work, that is
  a program c annotated by a Dijkstra monad type D a w0 is ``verified'' by
  computing the wp wc of c and then proving w0 <= wc}
\ch{This could be a place to explain this, for instance by explaining what
  happens if the user were to annotate $\stmod$ with a weaker spec.}
\ch{And if this is too silly an example for this, let's do it for pyths below.}
\da{I added a sentence to $\stmod$ about this.}

\newcommand{\NDd}{\mathrm{ND_\maltese}}
\newcommand{\NDD}[2]{\NDd\,{#1}\,{#2}}

\paragraph{Demonic nondeterminism} The previous construction is 
independent from how the computational monad,
the specification monad, and the effect observation were obtained. The exact same approach can be followed for the $\NDet$ monad coupled with any
of its effect observations. We use the demonic one here, for which
the $\choice$ and $\fail$ actions for the Dijkstra monad have types:
\begin{align*}
    \choice^{\NDd} \quad &: \quad \NDD{\bool}{(\abs{}{p}{p~\true \wedge p~\false})} \qquad\qquad
    \fail^{\NDd} \!\!\!&: \quad \NDD{\zero}{(\abs{}{p}{\top})}
\end{align*}
%
With this, we can define and verify \fstar{} (or Coq) functions like the following:
\begin{lstlisting}
  let rec pickl (l:list 'a) : NDD 'a \,(fun p -> forall x. elem x l ==> p x) =
    match l with | [] -> fail () | x::xs -> if pick () then x else pickl xs
  let guard (b:bool) : NDD unit (fun p -> b ==> p ()) = if b then () else fail ()
\end{lstlisting}
The \ls{pickl} function nondeterministically chooses an element from a list,
guaranteeing in its specification that the chosen value belongs to it. The
\ls{guard} function checks that a given boolean condition holds, failing otherwise.
The specification of \ls{guard b} ensures that \ls{b} is true
in the continuation.\ch{Do we really need to show and explain the specs of pickl
  and guard? Isn't it more impressive to just say that they are inferred?}
Using these two functions, we can write and verify concise
nondeterministic programs, such as the one below that computes Pythagorean triples.
The specification simply says that every result (if any!)
is a Pythagorean triple, while in the implementation we have
some concrete bounds for the search:
\begin{lstlisting}
let pyths () : NDD (int & int & int) (fun p -> forall x y z. x*x + y*y = z*z ==> p (x,y,z)) =
  let l = [1;2;3;4;5;6;7;8;9;10] in let (x,y,z) = (pickl l, pickl l, pickl l) in guard (x*x + y*y = z*z); (x,y,z)
\end{lstlisting}

\paragraph{Input-Output} We illustrate 
Dijkstra monads for multiple effect observations from $\IO{}$.
%
First, we consider the context-free interpretation
${\theta^{\Fr} : \IO \to \W^{\Fr}}$,
for which $\IO$ operations have the interface:
\begin{align*}
  \ioread^{\IOFr}  & :\;\;   \IOFR{I}{(\abs{}{p}{\forall (i:I).\, p\,\langle i, [\In~i]\rangle})}\\[-1.5mm]
  \iowrite^{\IOFr} & :\;\; (o:O) \rightarrow
                          \IOFR{\one}{(\abs{}{p}{p\,\langle \ast, [\Out~o]\rangle})}
\end{align*}
We can define and specify a program that duplicates its input (assuming an
implicit coercion $ I <: O$):
\begin{lstlisting}
let duplicate () : IOFree unit (fun p -> forall x. p ((), [In x; Out x; Out x])) = let x=input() in output x; output x
\end{lstlisting}

However, with this specification monad, we cannot reason
about the history of \textit{previous} IO events.
%
\newcommand{\WHist}{\W^\Hist}
To overcome this issue, we can switch the specification
monad to $\WHist$ and obtain
%
%
\newcommand{\IOHist}{\IO^{\Hist}}
\newcommand{\IOHIST}[2]{\mathrm{IOHist}\,{#1}\,{#2}}
\begin{align*}
    \ioread^{\IOHist} &:\;\; \IOHIST{I}{(\abs{}{p\,h}{\forall i.\, p\,\langle i, [\In~i] \rangle})} \\[-1mm]
    \iowrite^{\IOHist} &:\;\; (o:O) \rightarrow \IOHIST{\one}{(\abs{}{p\,h}{p\,\langle \ast, [\Out~o] \rangle})}
\end{align*}
%
The computational part of this Dijkstra monad fully coincides with
that of $\IOFr$, but the specifications are much richer.
For instance, we can define the following computation:
\[
    \begin{array}{rclcl}
        \texttt{mustHaveOccurred} & = & \abs{}{\_}{\ret^{\IOHist}\,\ast} &:& (o:O) \rightarrow \IOHIST{\one}
                                            {(\abs{}{p\,h}{\Out~o \in h \wedge p\, \langle \ast, []\rangle})} \\
    \end{array}
\]
which has no computational effect, yet requires that a given value $o$ was already been outputted before
it is called.
%
This is \emph{weakening}
the specification of $\ret^{\IOHist}~\ast$
(namely, $\ret^{\WHist}~\ast = \abs{}{p\,h}{p\,\langle \ast, [] \rangle}$)
to have a stronger
precondition.
%
By having this amount of access to the history, one can verify that certain
invariants are respected.
For instance, the following program will verify successfully:
\begin{lstlisting}
let print_increasing (i:int) : IOHist unit (fun p h -> forall h'. p ((), h')) =
    output i; (* pure computation *) mustHaveOccurred i; (* another pure computation *) output (i+1)
\end{lstlisting}
The program has a ``trivial'' specification: it does not guarantee
anything about the trace of events, nor does it put restrictions on the previous log.
However, internally, the call to \texttt{mustHaveOccurred} has a
precondition that \ls{i} was already output, which can be proven from
the postcondition of \ls{output i}.
If this \ls{output} is removed, the program will (rightfully) fail
to verify.

\newcommand{\IOST}[2]{\mathrm{IOSt}\,{#1}\,{#2}}

Finally, when considering the specification monad $\WIOSt$,
we have both state and IO operations:
\begin{align*}
    \ioread^{\IOSt} &: \IOST{I}{(\abs{}{p\, s\,h}{\forall i.\, p\,\langle i, s, [\In~i] \rangle})} 
    &
    \stget^{\IOSt} &: \IOST{S}{(\abs{}{p\, s\,h}{p\,\langle s, s, [] \rangle})} \\[-1mm]
    \iowrite^{\IOSt} &: (o:O) \rightarrow \IOST{\one}{(\abs{}{p\,s\,h}{p\,\langle \ast, s, [\Out~o] \rangle})}
    &
    \stput^{\IOSt} &: (s:S) \rightarrow \IOST{\one}{(\abs{}{p\,\_\,h}{p\,\langle \ast, s, [] \rangle})}
\end{align*}
where $(\ioread^{\IOSt},\iowrite^{\IOSt})$ keep state unchanged, and $(\stget^{\IOSt},\stput^{\IOSt})$
do not perform any IO. With this, we can write and verify programs that combine state and
IO in non-trivial ways, \EG
\begin{lstlisting}
let do_io_then_rollback_state () : IOST unit (fun s h p -> forall i . p (() , s , [In i; Out (s+i+1)])) =
  let x = get () in let y = input () in put (x+y); (* pure computation *) let z = get () in output (z+1); put x
\end{lstlisting}
The program mutates the state in order to compute output from input, possibly 
interleaved with pure computations, but eventually
rolls it back to its initial value, as mandated by its specification.



\paragraph{Effect polymorphic functions}
Even though the operations $\ret$ and $\bind$ provided by a (strong) monad
can seem somewhat restrictive at first,\ch{no they don't; no clue what restriction
  you're talking about; writing useful code with {\bf only} return and bind?}\da{a generic list 
  map function is not useful?}
they still allow us to write functions that
are generic in the underlying computational monad. One example is the following \ls{mapW}
function on lists, generic in the monad \ls{W} (similar to the
\ls{mapM} function in Haskell):
\begin{lstlisting}
let rec mapW (l : list 'a) (f : 'a -> W 'b) : W (list 'b) =
  match l with [] -> ret [] | x :: xs -> bind (f x) (fun y -> bind (mapW xs f) (fun ys -> ret (y :: ys)))
\end{lstlisting}
When working with Dijkstra monads, we can use the \ls{mapW} function 
as a generic specification for the same computation when expressed 
using an arbitrary Dijkstra monad
\ls{D} indexed by \ls{W}:\footnote{These \iffull\else last \fi
   examples are written in
  \fstar{} syntax, but only implemented in Coq, since Dijkstra monads
  are not first class in \fstar{}.}
\begin{lstlisting}
let rec mapD (l : list 'a) (w : 'a -> W 'b) (f : (a:'a) -> D 'b (w a)) : D (list 'b) (mapW l w) =
  match l with [] -> ret [] | x :: xs -> let y = f x in let ys = mapD xs w f in y :: ys
\end{lstlisting}
where \ls{mapD} takes the list \ls{l}, the specification for what is
to happen to each element of the list, \ls{w}, and an implementation
of that specification, \ls{f}. It builds an effectful computation that produces
a list, specified by the extension of the element-wise specification
\ls{w} to the whole list by \ls{mapW}.
%

Analogously, we can implement a generic iterator combinator
provided we have an invariant \ls{w : W unit} for the loop
\ls{body : nat -> D unit w} such that the invariant satisfies 
\ls{bind w (fun () -> w) <= w}\bob{is this the right way round?}\da{if weakening 
is contravariant like in \autoref{sec:dijkstraMonads}, then no 
(seems to have been since flipped around)}\da{Kenji has since flipped the 
direction of orders in the paper (to also match the Coq development)}:
\begin{lstlisting}
let rec for_in (range : list nat) (body : nat -> D unit w) : D unit w =
  match range with [] -> ()  | i :: range -> body i ; for_in range body
\end{lstlisting}
Here we use not only the monadic operations but also the
possibility to weaken the specification \ls{bind w (fun () -> w)}
computed from the second branch of the \ls{match} to the specification
\ls{w} by assumption.

\ch{What's not at clear here is that these 2 examples are verified
  only using the monad laws. In general, showing all these examples
  without explaining at least a bit how they are verified is a bit of a waste.}
\ch{So better to keep a single example here, but explain it better?}

In all the examples in this section, we used Dijkstra monads obtained
via the same general recipe (see \autoref{sec:dijkstraMonads} for details) from the same kinds
of ingredients: a computational monad, a specification monad, and an
effect observation from the former to the latter.
This enables a uniform treatment of effects for verification, and opens
the door for verifying rich properties of effectful programs.

\section{Defining Specification Monads}
\label{sec:metalanguage}

To enable various verification styles, in \autoref{sec:overview} we
introduced various specification monads arising from the application
of monad transformers to the monad of predicate transformers $\W^{\Pure}$.
%
%
In this section, we start by observing that $\W^{\Pure}$ is not the only 
natural basic specification monad on which to stack monad 
transformers (\autoref{sec:basic-spec-monads}).
We then present our \emph{specification metalanguage} \SM{},
as a means for defining correct-by-construction monad transformers
(\autoref{sec:monad-transformers}).
\SM{} is a more principled variant of the \DM{} language of
\citet{dm4free}, and similarly to \DM{}, we give \SM{} a semantics
based on logical relations.
Observing that not all \SM{} terms give rise to monad transformers
(\autoref{sec:continuationMonadTrans}), we extract conditions under
which we are guaranteed to obtain monad transformers, providing an
explanation for the somewhat artificial syntactic restrictions in \DM{}.
%
Finally, we also discuss a principled way to derive effect
observations into $\W^{\Pure}$ and $\W^{\St}$ from  
algebras of computational monads (\autoref{sec:monad-algebras}).


\subsection{Basic Specification Monads}
\label{sec:basic-spec-monads}

We consider several basic specification monads,
whose relationship is summarized by \autoref{fig:specmonads}.
%


\paragraph{Predicate monad}
Arguably the simplest way to specify a computation is to provide a
postcondition on its outcomes.
This can be done by considering the specification monad 
$ \pred\,A \eqdef A \to \prop $ (the covariant powerset monad) with order $p_{1} \leq^{\pred} p_{2} \iff 
\forall (a:A).\, p_{1}\,a \Rightarrow p_{2}\,a$.
To specify the behavior of returning values, we can
always map a value $v:A$ to the singleton predicate 
$\ret^{\pred}\,v \eqdef \abs {} y {(y = v)} : \pred\,A$. 
And given a predicate $p : \pred\,A$ and a
function $f : A \to \pred\,B$, the predicate on $B$ defined by $\bind^{\pred}\,p\,f\eqdef\abs{}b{\exists a.\, p\,a \wedge f\,a\,b}$ 
specifies the behavior of sequencing two computations, where the
first computation produces a value $a$ satisfying $p$ and, under this assumption, the second computation produces a value satisfying $f\,a$.
While a specification $p:\pred\,A$ provides information on
the outcome of the computation, it cannot require preconditions, so
computations need to be defined independently of any logical context.
To give total correctness specifications to computations with
non-trivial preconditions, for instance
specifying that the division function $\mathrm{div}\,x\,y$ requires $y$ to be
non-zero, we need more expressive specification monads.
%
%

%
\paragraph{Pre-/postcondition monad}
One more expressive specification monad is\da{was ''we study'';
  Is this a novel contribution of ours?}\km{From discussions last
  year with P.E.Dagand I think it's already known to people using characteristic
formula, but I don't know any reference}\ch{To me this seems closely related
  to Ynot~\cite{ynot-icfp08} and HTT~\cite{nmb08htt}. TODO: I would be very
  curious whether we can recover a version of Ynot by applying the State +
  Exceptions (+ Partiality) to this basic specification monad.}\ch{I am in
  fact very curious what happens when one applies monad transformers
  to this base specification monad; can one always re-factor things
  again in terms of pre- and postconditions? It would be interesting
  to show at least some examples, like State + Exceptions, and maybe
  also prove a general result about this.}\km{I never managed to get such a
  factorization but I would be extremely interested in having one}
the monad of pre- and postconditions 
${\prepost\,A \eqdef \prop \times (A \to \prop)}$, bundling
a precondition together with a postcondition.
%
Here the behavior of returning a value $v:A$ is specified by requiring
a trivial precondition and  ensuring as above a singleton postcondition:
$\ret^{\prepost}\,v \eqdef \pair{\top}{\abs {} a {a = v}} : \prepost\,A$.
And, given $p = \pair{pre}{post} : \prepost\,A$ and a function $f = \abs{}{a}{\pair{pre'~a}{post'~a}} : A \to
\prepost\,B$, the sequential composition of two computations 
is naturally specified by defining
\[\bind^{\prepost}\,p\,f\eqdef
  \pair{~\left(pre \wedge \forall a.\, post\, a \implies pre'~a\right)~}
  {~\abs{}b{\exists a.\, post\,a \wedge
    post'~a\,b}~} : \prepost\,B\]
The resulting precondition ensures that the precondition of the first computation
holds and, assuming the postcondition of the first computation,
the precondition of the second computation also holds.
The resulting postcondition is then simply the conjunction of the
postconditions of the two computations.
The order on $\prepost{}$ naturally combines the 
pointwise forward implication order on postconditions with 
the backward implication order on preconditions. 

We formally show that this specification monad is more expressive
than the predicate monad above:
Any predicate $p : \pred\,A$ can be coerced to $(\top, p) :
\prepost\,A$, and in the other direction, any pair $(pre, post) :
\prepost\,A$ can be approximated by the predicate $post$, 
giving rise to a Galois connection, as illustrated in
\autoref{fig:specmonads}.\ch{Could return to the div example to explain why
  this is an overapproximation, but it would take more space.}
While the monad $\prepost{}$ is intuitive for humans, generating
efficient verification conditions is generally easier for predicate
transformers~\cite{Leino05}.
%

\paragraph{Forward predicate transformer monad}
The predicate monad $\pred$ can be extended in an alternative way.
Instead of fixing a precondition as in $\prepost{}$, a specification
can be a function from preconditions to postconditions, for instance
producing the strongest postcondition of computation for any
precondition $\pre {:} \prop$ given as argument.
%
%
%
Intuitively, such a forward predicate transformer on $A$ has type
$\prop \to (A \to \prop)$.
However, to obtain a monad 
(\IE~satisfying the expected laws),
we have to consider the smaller type
$\strpost\,A \eqdef (\pre{:}\prop) \to (A \to \prop_{/\pre})$ of
predicate transformers 
monotonic with respect to $\pre$, where
$\prop_{/\pre}$ is the subtype of propositions implying $\pre$.
Returning a value $v : A$ is specified by the predicate transformer  
$\ret^{\strpost}\,v \eqdef \abs{}{\pre\,a}{\pre \wedge a = v}$, and
the sequential composition of two computations 
is specified as the predicate transformer $\bind^{\strpost}\,m\,f \eqdef \abs{}{\pre\,b}{\exists a.\,
  f\,a\,(m\,\pre\,a)\,b}$, for $m : \strpost\,A$ and $f : A \to \strpost\,B$.
%

%
\paragraph{Backward predicate transformer monad}
As explained in~\autoref{sec:predicateTransformers}, 
backward predicate transformers
can be described using the
continuation monad with propositions $\prop$ as the answer type, namely, 
${\Cprop\,A \eqdef (A \to \prop) \to \prop}$.
Elements $w : \Cprop\,A$ are predicate transformers mapping a
postcondition $\post : A \to \prop$ to a precondition
$w\,\post : \prop$, for instance the weakest precondition of the
computation.
Pointwise implication is a natural order on $\Cprop\,A$:
\[ w_{1} \leq w_{2} : \Cprop\,A \quad \Leftrightarrow \quad
  \forall (p : A \to \prop).\,  w_{2}\, p \Rightarrow w_{1}\,p  \]
%
However, $\Cprop$ is not an ordered monad with respect to this
order because its $\bind$ is not monotonic.
In order to obtain an ordered monad, we restrict our attention
to the submonad $\W^\Pure$ of $\Cprop$ containing the {\em monotonic}
predicate transformers, that is those $w : \Cprop\,A$ such that
\[ \forall (p_{1}\,p_{2} : A \to \prop). \quad (\forall (a:A).\, p_{1}\,a
  \Rightarrow p_{2}\,a) \quad \Rightarrow \quad w\,p_{1} \Rightarrow w\,p_{2},\]
%
%
which is natural in verification:
we want stronger postconditions to map to
stronger preconditions.

\ch{One point that's missing is the reason one chooses wps over all
  other base specification monads in practice: its bind doesn't
  introduce any quantifiers, while all the 3 monads above used
  quantifiers for defining bind. Need to show the definition of bind
  to make this point though.}

This specification monad is more expressive than the
pre-/postcondition one above~\cite{mumon}.\ch{more authoritative references?}
Formally, a pair $(\pre, \post) :\prepost\,A$ can be mapped
to the predicate transformer
\[\abs{} {(p : A \to \prop)} {\pre \wedge (\forall (a:A).\, \post\,a
    \Rightarrow p\,a)} \quad : \quad \W^\Pure\,A,\]
and vice versa, a predicate transformer $w:\W^\Pure\,A$ can
be approximated
by the pair
\[ (\quad w\,(\abs{}a\top)\quad, \quad \abs{} a {(\forall p.\, w\,p
    \Rightarrow p\,a)}\quad) \quad : \quad \prepost\,A\]
These two mappings define a Galois connection, as illustrated in \autoref{fig:specmonads}. 
Further, this Galois connection exhibits $\prepost\,A$
as the submonad of $\W^\Pure\,A$ of {\em conjunctive}
predicate transformers, \IE predicate transformers $w$
commuting with non-empty conjunctions/intersections.
\label{sec:PrePostWPs}

Finally, both $\W^\Pure$ and $\strpost$ can be embedded into an even
more expressive specification monad $\relprepost$ consisting of
relations between preconditions and postconditions satisfying a few
conditions, the full details of which can be found in our Coq
formalization.
%

%
%
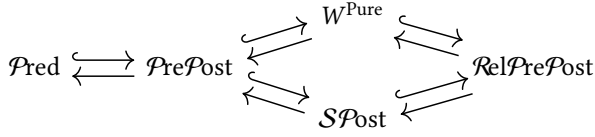
\begin{figure}
\[
  \begin{tikzcd}[row sep=tiny]
    & & \W^\Pure \ar[dl,shift left=0.3em] \ar[hookrightarrow, dr, shift left=0.3em]& \\
    \pred \ar[hookrightarrow, r,shift left=0.3em] & 
    \prepost \ar[l,shift left=0.3em] \ar[hookrightarrow, ur,shift left=0.3em]
    \ar[hookrightarrow,dr,shift left=0.3em]&  &
    \relprepost \ar[ul,shift left=0.3em] \ar[dl,shift left=0.3em] \\
    & & \strpost \ar[ul,shift left=0.3em] \ar[hookrightarrow, ur,shift left=0.3em] & \\
  \end{tikzcd}
\]
\vspace{-1.5em}
\caption{Relationships between basic specification monads (each pair forms a
  Gallois connection)}
\label{fig:specmonads}
\end{figure}

\subsection{Defining Monad Transformers}
\label{sec:monad-transformers}

We use \emph{monad transformers}~\cite{LiangHJ95} to construct more
complex specification monads from the basic ones above (and in some
cases also to derive effect observations~\autoref{sec:EffObsExamples}).
However, defining a monad transformer and proving that
it satisfies all the expected laws requires significant effort.
In this section, we introduce a \emph{Specification Metalanguage},
\SM{}, and a translation from \SM{} to correct-by-construction monad
transformers in a base dependent type theory $\bL$ (where $\bL$ is a
parameter of \SM).
More precisely, our translation takes as input a monad in \SM{}
subject to two extra conditions, \emph{covariance} and
\emph{linearity}, and produces a correct monad transformer
in $\bL$.
\newcommand{\SMRd}{\mathbb{Rd}}
\newcommand{\SMWr}{\mathbb{Wr}}
\newcommand{\SMExc}{\mathbb{Exc}}
\newcommand{\SMSt}{\mathbb{St}}
\newcommand{\MonSMSt}{\mathbb{MonSt}}
\newcommand{\SMCont}[1]{\mathbb{Cont}_{#1}}
\newcommand{\Ans}[0]{\mathcal{A}\mathrm{ns}} 
\SM{} is an expressive language in which many different monads can be
defined in a natural way, for example
%
%
%
\emph{reader} $\SMRd(X :\Type) \eqdef \mathcal{I} \to \bM\,X$; 
\emph{writer} $\SMWr(X:\Type) \eqdef \bM (X \times \mathcal{O})$; 
\emph{exceptions} $\SMExc(X:\Type) \eqdef \bM (X + \mathcal{E})$; 
\emph{state} $\SMSt(X:\Type) \eqdef \mathcal{S} \to \bM (X \times \mathcal{S})$;
\emph{monotonic state}
  $\MonSMSt(X)\eqdef (s_{0} : \mathcal{S}) \to \bM (X \times (s_{1} :
  \mathcal{S}) \times s_{0} \preccurlyeq s_{1})$, where $\preccurlyeq$ is some
  preorder on states $\mathcal{S}$; and
\emph{continuations} $\SMCont{\Ans}(X)
  \eqdef (X \to \bM\,\Ans) \to \bM\,\Ans$.
The symbol $\bM$ stands for an arbitrary base monad, and
the \emph{covariance} condition states that it appears only in the codomain of
arrows. The more involved \emph{linearity} condition concerns the $\bind$
of these monads.
With the exception of continuations (see \autoref{sec:continuationMonadTrans}), all these \SM{} monads satisfy these
extra conditions %
and thus lead to proper %
monad transformers.
%

%
\begin{definition}[Monad transformer] A monad transformer~\cite{LiangHJ95} is given by 
\begin{itemize}[nosep,label=$\vartriangleright$,leftmargin=*]
\item a function $\mathcal{T}$ mapping monads $\M$ to monads $\mathcal{T}\M$, 
\item equipped with a monad morphism $\lift_{\M} : \M \to \mathcal{T}\M$, 
\item assigning functorially to each monad morphism $\theta : \M_1 {\to} \M_2$
  a monad morphism ${\mathcal{T}\theta : \mathcal{T}\M_1{\to} \mathcal{T}\M_2}$,  
\item and such that the $\lift_{\M}$ is natural in $\M$, that is for any monad
  morphism $\theta : \M_{1} \to \M_{2}$,
  \[\mathcal{T}\theta \circ \lift_{\M_{1}} = \lift_{M_{2}} \circ \theta \]
\item moreover, they need to preserve the order structure present
on the (ordered) monads as well as the monotonicity of morphisms, and the lifts themselves should also be monotonic, 
\end{itemize}
i.e., $(\mathcal{T},\lift)$ is a pointed endofunctor on the
category of (ordered) monads~\cite{LuthG02}. 
%
\end{definition}
%
%
%
%
\begin{figure}
  \begin{mathpar}
    C ::= \bM A \alt C_1 \times C_2 \alt (x:A) \to C \alt C_1 \to C_2 \quad A
    \in \ii{Type}_{\mathcal{L}}
    \vspace{-0.2cm}
    \\
    t ::= \ret \alt \bind \alt \pair \tone \ttwo \alt \proj i\; t \alt x \alt
    \abs {\diamond} x t \alt \app \tone \ttwo \alt \abs {} x t \alt \app t u \quad u \in \ii{Term}_{\mathcal{L}}
  \end{mathpar}
  \vspace{-2em}
  \caption{Syntax of \SM{}}
  \label{fig:DMDef}
  \vspace{-1em}
\end{figure}
\begin{figure}
  \begin{mathpar}
    \inferrule{ }{A \vdD \ret~:~ A \to \bM A}
    \and
    \inferrule{ }{A,B \vdD \bind~:~\bM A \to (A \to \bM B) \to \bM B}
    \and
    \inferrule{\Gamma, x:C_{1} \vdD t : C_{2}}{\Gamma \vdD \abs \diamond x t : C_{1} \to C_{2}}
  \end{mathpar}
  \vspace{-1.2em}
  \caption{Selected typing rules for \SM{}}
  \label{fig:DMTyping}
\end{figure}
\begin{figure}
  \begin{mathpar}
    \elab \M {\bM A} = \M A \and \elab \M {C_{1} \times C_{2}} = \elab \M
    {C_{1}} \times \elab \M {C_{2}} \and \elab \M {(x:A) \to C} = (x:A) \to
    \elab \M C 
    \vspace{-0.15cm}
    \\
    \elab \M {C_{1} \to C_{2}} = (f : \elab \M {C_{1}} \to \elab
    \M {C_{2}}) \times (\forall (m_{1} \leq^{C_{1}} m'_{1}).\, f\,m_{1} \leq^{C_{2}} f\,m'_{1})
    \vspace{-0.15cm}
    \\
    m \leq^{\bM\,A} m' \eqdef m \leq^{\M}_{A} m'
    \and
    \pair{m_{1}}{m_{2}} \leq^{C_{1} \times C_{2}} \pair{m'_{1}}{m'_{2}} \eqdef
    m_{1} \leq^{C_{1}} m'_{1} \wedge m_{2} \leq^{C_{2}} m'_{2}
    \vspace{-0.15cm}
    \\
    f \leq^{(x:A) \to C[x]} f' \eqdef \forall (x :A).\, f\,x \leq^{C[x]} f'\,x
    \and
    f \leq^{C_{1} \to C_{2}} f' \eqdef \forall (m_{1}\leq^{C_{1}} m'_{1}) 
    .\, f\,m_{1} \leq^{C_{2}} f'\,m'_{1}
  \end{mathpar}
  \vspace{-2em}
  \caption{Elaboration from \SM{} to $\bL{}$}
  \label{fig:ElabDM}
\end{figure}
\paragraph{Building monad transformers}
The design of \SM{}, whose syntax is presented in \autoref{fig:DMDef},
has been informed by the goal of defining monad transformers.
First, since we want a mapping from monads to monads, we introduce
the type constructor $\bM$ standing for an arbitrary base monad, as well
as terms $\ret$ and $\bind$.
Second, in order to describe monads internally to \SM{}, we add function types
$(x:A) \to C[x]$ and $C_{1} \to C_{2}$.
We allow dependent function types only when the domain is in $\bL$, leading to
two different type formers.
We write dependent abstractions as $\abs{}xt$, whereas we
write the non-dependent type as $\abs{\diamond}xt$. 
In \autoref{fig:DMTyping} we present the typing rules of $\bM$, $\ret$, and 
$\bind$, leaving the remaining standard \SM{} typing rules for
\ifappendix
\autoref{sec:linearTyping}.
\else
the online appendix.
\fi
To define our monad transformers, we use monads {\em internal to} \SM{}, given by
%
\begin{itemize}[leftmargin=*,nosep,label=$\vartriangleright$]
\item a type constructor $X:\Type \vdD C[X]$;
\item terms $A : \Type \vdD \ret^{C} : A \to C[A]$ and ${A,B:\Type
  \vdD \bind^{C}:(A \to C[B]) \to C[A] \to C[B]}$;
\item such that the monadic laws are derivable in the
  equational theory of \SM{}.
\end{itemize}

Now, given a monad $C$ internal to \SM{}, we want to define the corresponding
monad transformer $\mathcal{T}^C$ evaluated at a monad $\M$ in the base language $\bL$, 
essentially as the substitution of $\M$ for $\bM$.
\ch{Starting around here things get too brutal and dry for me to follow.}%
\ch{How about illustrating each of the generic constructions
  for at least one of the example SM monads?}\ch{It's getting better,
  but even better would be to actually relate the general definitions and
  the examples.}%
In order to make this statement precise,
%
 %
%
we define a denotation $\elab{\M}{-}^{(\gamma)}$ in $\bL$ of \SM{} types
(\autoref{fig:ElabDM}) and terms (provided in the appendix
together with the equational theory) parametrized by $\M$.
This denotation preserves the equational theory of \SM{}, provided $\bL$ has
extensional dependent products and pairs.
As such, $C$ induces the following mapping from monads to monads:
\[ \mathcal{T}^C \quad:\quad (\M, \ret, \bind) \qquad \longmapsto \qquad (\llbracket C \rrbracket_{\M}, \llbracket
  \ret^{C} \rrbracket_{\M}, \llbracket \bind^{C} \rrbracket_{\M})\]
For instance, taking $C = \SMSt$, the definition evaluates to $\mathcal{T}^\SMSt\M = X \mapsto \mathcal{S} \to M (X \times \mathcal{S})$.
To build the $\lift$ for $\mathcal{T}^C$, the key observation
is that the denotation $\llbracket C
\rrbracket_{\M}$ of an \SM{} type $C$ in $\bL$ can be endowed with an
$\M$-algebra structure $\alpha^{C}_{\M} : \M\elab{\M}{C} \to \elab{\M}{C}$\footnote{
  An \emph{$M$-algebra} is an object $X$ together with a map $\alpha : MX \longrightarrow X$, 
  which is required to respect $\retM$ and $\bindM$.}.
This $\M$-algebra structure is defined by induction on the
structure of the \SM{} type $C$, using the free algebra when $C =
\bM A$ and the pointwise defined algebra in all the other cases.
This $\M$-algebra structure allows us to then define a lifting function
from the monad $\M$ to the monad $\llbracket C \rrbracket_{\M}$
as follows:
\[\lift^{C}_{\M, X} \quad:\quad
  \M(X)
  \xrightarrow{\M(\ret^{\elab{\M}{C}})}
  \M\elab{\M}{C}(X)
  \xrightarrow{\alpha^{C}_{\M,X}}
  \elab{\M}{C}(X) = \mathcal{T}^C\M(X)
  \]
For instance, $\lift^{\SMSt}_{\M,X}\,(m:\M\,X) =
\abs{}{(s:\mathcal{S})}{\M(\abs{}{(x:X)}{\langle {} x, s \rangle})\,m} :
\mathcal{S} \to \M(X \times S)$.
The result that \SM{} type formers are automatically equipped with an
algebra structure explains why \SM{} features products, but not sums since the
latter cannot be equipped with an algebra structure in general.
This $\lift^{C}_{\M} : \M \to \elab{\M}{C}$ needs to be \emph{natural}, that is, the
following diagram should commute: 
\[
  \begin{tikzcd}[column sep = huge]
    \M\,A \ar[r, "\M(\ret^{C}_{A})", ""'{name=srcl}] \ar[d, "\M\,f"'] &\M\elab{\M}{C}\,A \ar[r,
    "\alpha^{C}_{\M,A}",""'{name=srcr}] \ar[d, "\M\elab{\M}{C}\,f"] & \elab{\M}{C}\,A \ar[d, "\elab{\M}{C}\,f"]\\
    \M\,B \ar[r, "\M(\ret^{C}_{B})"', ""{name=dstl}] &\M\elab{\M}{C}\,B \ar[r,
    "\alpha^{C}_{\M,B}"',""{name=dstr}] & \elab{\M}{C}\,A
    \ar[from=srcl, to=dstl, phantom, "\circlearrowleft"]
    \ar[from=srcr, to=dstr, phantom, "?"]
  \end{tikzcd}
\]
for any $A,B$ and $f : A \to B$.
The left square commutes automatically by the naturality of $M(\ret^{C})$. For the right
square to commute, however, 
$\elab{\M}{C}\,f = \bind^{\elab{\M}C} (\ret^{\elab{\M}C} \circ f)$ should be an $\M$-algebra homomorphism.
We can ensure it by asking that $\bind^{C}$ maps functions to $\M$-algebra
homomorphisms, a condition that can be syntactically captured by a
linearity condition in a modified type system for \SM{} equipped with a {\em stoup},
which is a distinguished variable in the context such that the term is linear
with respect to that variable \cite{EggerMS14,MunchMaccagnoni13}.
We omit this refined type system here and refer to 
\ifappendix
\autoref{sec:linearTyping}
\else
the online appendix
\fi
for the complete details.
%
%
We call this condition on the monad $(C,\ret^{C},\bind^{C})$ internal to \SM{} \emph{the linearity of $\bind^{C}$}. 
%

%
\paragraph{Action on monad morphism}
To define a monad transformer, we still need to build a
functorial action mapping monad morphism $\theta : \M_{1} \to
\M_{2}$ between monads $\M_{1}, M_{2}$ in $\bL$ to a monad
morphism $\elab {\M_{1}} C \to \elab {\M_{2}} C$.
However, the denotation of the arrow $C_{1} \to C_{2}$ does not
allow for such a functorial action since $C_1$ necessarily contains a subterm
$\bM$ in a contravariant position.\ch{Why not? Please explain.}
In order to get an action on monad morphisms, we first build a
(logical) relation between the denotations. Given $\M_{1}, \M_{2}$
monads in $\bL$ and a family of relations
$R_{A} \subset \M_{1}A \times \M_{2}A$
indexed by types $A$, we build a relation
$\elabrel {\M_{1}} {\M_{2}}{R}C \subset \elab {\M_{1}} C\times
\elab {\M_{2}} C$ as follows
\begin{align*}
  m_{1} ~ \elabrelbase{\bM A} ~ m_{2} \qquad&\eqdef\qquad m_{1}~ R_{A} ~ m_{2}\\
  (m_{1}, m'_{1}) ~ \elabrelbase{C_{1} \times C_{2}} ~ (m_{2}, m'_{2}) \qquad&\eqdef\qquad
  m_{1} ~ \elabrelbase{ C_{1} } ~ m_{2} \wedge m'_{1} ~ \elabrelbase{ C_{2}
  } ~ m'_{2}\\
  f_{1} ~ \elabrelbase{ (x:A) \to C } ~ f_{2} \qquad&\eqdef\qquad \forall (x:A).\, f_{1}\;x ~ \elabrelbase{
  C\; x } ~ f_{2}\;x\\
  f_{1} ~ \elabrelbase{ C_{1} \to C_{2} } ~ f_{2} \qquad&\eqdef\qquad \forall m_{1}\, m_{2}.\,
  m_{1} ~ \elabrelbase{ C_{1} } ~ m_{2} \Rightarrow f_{1}\; m_{1} ~ \elabrelbase{ C_{2}
  } ~ f_{2}\;m_{2}
\end{align*}
Now, when a type $C$ in \SM{} comes with the data of an internal
monad, the relational denotation $\elabrel{\M}{\W}{-}{C}$ maps not
only families of relations to families of relations, but also preserves
the following structure that we call a \emph{monadic relation}:
\begin{definition}[Monadic relation]
  \label{def:monadicRelation}
A monadic relation $\mathcal{R} : \M \leftrightarrow
  \W$ between a computational monad $\M$ and a specification
  monad $\W$, consists of:
  \begin{itemize}[nosep,label=$\vartriangleright$]
  \item a family of relations $\mathcal{R}_{A} : \M\,A
    \times \W\,A \to \prop$ indexed by type $A$
  \item such that returned values are related
    $ 
    (\ret^{\M}\,v)\,\mathcal{R}_{A}\,(\ret^{\W}\,v)$ for any
    value $v : A$
  \item and such that sequencing of related values is related
    \begin{mathpar}
      \inferrule{
        m_{1} \,\mathcal{R}_{A}\, w_{1}\\
        \forall (x:A).\,  (m_{2}\,x)\,\mathcal{R}^{B}\, (w_{2}\,x)}
      {(\bindM\;m_{1}\;m_{2})\,\mathcal{R}^{B}\, (\bindW\;w_{1}\;w_{2})}
    \end{mathpar}
  \end{itemize}
\end{definition}
The simplest example of monadic relation is the graph of a monad morphism
$\theta : \M \to \W$.
Given a monadic relation, we extend the relational translation to terms and
obtain the so-called fundamental lemma of logical relations.
\begin{theorem}[Fundamental lemma of logical relations]
  For any monads $\M_{1}, \M_{2}$ in $\bL$, monadic
  relation $\mathcal{R} : \M_{1} \leftrightarrow \M_{2}$, term $\Gamma \vdD t :
  C$ and substitutions $\gamma_{1} : \elab{\M_{1}}{\Gamma}$ and $\gamma_{2} :
  \elab{\M_{2}}{\Gamma}$, if for all $(x:C') \in \Gamma$,
  $\gamma_{1}(x)~\elabrel{\M_{1}}{\M_{2}}{\mathcal{R}}{C'}~\gamma_{2}(x)$ then $\elab{\M_{1}}{t}^{\gamma_{1}}~\elabrel{\M_{1}}{\M_{2}}{\mathcal{R}}{C}~\elab{\M_{2}}{t}^{\gamma_{2}}$.
\end{theorem}
As a corollary, an internal monad $C$ in \SM{} preserves monadic
relations, the relational interpretation of $\ret^{C}$ and $\bind^{C}$ providing
witnesses to the preservation of the monadic structure.
In particular, any monad morphism $\theta : \M_{1} \to \M_{2}$ defines a monadic
relation $\elabrel{\M_{1}}{\M_{2}}{\mathrm{graph}(\theta)}{C} : \elab{\M_{1}}{C} \leftrightarrow
\elab{\M_{2}}{C}$.
%
It turns out that if $C$ is moreover \emph{covariant}, meaning that it does not
contain any occurrence of an arrow $C_{1} \to C_{2}$ where $C_{1}$ is a type in
\SM{}, then the relational denotation
$\elabrel{\M_{1}}{\M_{2}}{\mathrm{graph}(\theta)}{C}$ with respect to any  monad
morphism $\theta : \M_{1} \to \M_{2}$ is actually the graph of a monad morphism.
%
%
To summarize:
\begin{theorem}[Construction of monad transformer from \SM{}]
  \label{thm:MonadTransf}
\text{}
Given a monad $C$ internal to \SM{}
such that $\bind^{C}$ satisfies the linearity criterion, we obtain:
\begin{itemize}[leftmargin=*,label=$\vartriangleright$]
\item if $C$ is covariant, 
  then $\mathcal{T}^C$ equipped with $\lift^{C}_{\M}{:} \M {\to} 
  \mathcal{T}^C\M$ is a (ordered) monad transformer;
\item if $C$ is not covariant, $\mathcal{T}^C$ defines a pointed endofunctor on
  the category of (ordered) monads and monadic relations.\ch{Tried to make this
    more formal, but now missing the order? In general, the treatment
    of the order is quite handwavy in this whole section.}
\end{itemize}
\end{theorem}

We note that the resulting design for \SM{}
is close to Moggi's monadic metalanguage,
since it contains the same type formers: a unary type
former $\bM$, products $C_{1}\times C_{2}$ and functions
$A \to C$.
The main difference is that \SM{} is not parameterized on
a set of simple base types but on a dependent type theory
$\mathcal{L}$.
As such, \SM{} captures the essential elements of the metalanguage \DM{} of
\citet{dm4free}, leaving the non-necessary parts, such as sum
types, to the base language $\bL$.

\subsection{The Continuation Monad Pseudo-Transformer}
\label{sec:continuationMonadTrans}
%
Crucially, the internal continuation monad $\SMCont{\Ans}$ does \emph{not} verify
the conditions to define a monad transformer since it is not covariant in $\bM$.
We study this (counter-)example in detail since it extends the definition of
\citet{JaskelioffM10} to monadic relations and clarifies the prior work of
\citet{dm4free}, where a Dijkstra monad was obtained in a similar way.
%

While \SM{} gives us both the computational continuation monad $\elab \Id {\SMCont \Ans} =
\Cont_{\Ans}$ and the corresponding specification monad $ \elab \Cprop {\SMCont \Ans} =
\Cont_{\Cprop(\Ans)}$, we only get a monadic relation between the two and not a
monad morphism. We write this monadic relation as follows:
\[ \elab \Id {\SMCont \Ans} \quad
  \longleftarrow\!\!\elabrel \Id \Cprop \ret {\SMCont \Ans} \!\!\longrightarrow \quad
  \elab \Cprop {\SMCont \Ans}
\]
One probably wonders what are the elements related by this relation? Unfolding the
definition, we get that a computation $m : \elab \Id {\SMCont \Ans}(X)$ and a specification $w : \elab
\Cprop {\SMCont \Ans} (X)$ are related if
\begin{align*}
  &\;\;m~\elabrel \Id \Cprop \ret {\SMCont \Ans}~w\\
  &\Leftrightarrow \forall (k : X \to \Ans)\, (w_{k} : X \to \Cprop(\Ans)).\, (\forall (x:X).\, \ret\,(k\, x) = w_{k}\,x) \Rightarrow \ret\,(m\,k) = w\,w_{k}\\
  &\Leftrightarrow \forall (k : X \to \Ans).\, \ret\,(m\,k) = w\,(\abs {} x {\ret\,(k\,x)})\\
  &\Leftrightarrow \forall (k : X \to \Ans)\,(p:\Ans \to \prop).\, w\,(\abs {} {x\,q} {q\,(k\,x)})\,p = p\,(m\,k)
\end{align*}
For illustration, if we take $\Ans = \one$, the last condition
reduces to ${\forall (p:\prop).\, w\,(\abs{}{x\,q}q)\,p = p}$, in particular any
sequence $x_{0},\ldots, x_{n}$ induces an element $w = \abs{}{k\,p}{k\,
  x_{0}(\ldots k\, x_{n}\,p)} : \elab \Cprop {\SMCont \Ans} (X)$ that can be
seen as a specification revealing some intensional information about the
computation $m$ at hand, namely, that the continuation $k$ was called with the
arguments $x_{0}, \ldots, x_{n}$ in this particular order.
%
Computationally however, in the case of $\Ans = \one$, $m$ is
extensionally equal to $\abs {} k \ast : \Cont_{\one}$. \km{Can we specify
  computations using their continuation linearly?}

\ch{It would be great to step back and give a crisp summary what all
  this means for verifying programs that use continuations.
  And if we do this, we should also cite \citet{DelbiancoN13},
  for work that does serious verification of programs with
  continuations. Maybe something along the lines: scaling
  this idea to more interesting examples~\cite{DelbiancoN13}
  is future work.}
\ch{Now that we have this in Coq we should probably say something more
  about it?}
%
%
%
%
%

%

\subsection{Effect Observations from Monad Algebras}
\label{sec:monad-algebras}
While monad transformers $\mathcal{T}$ enable us to derive complex
specification monads, they can only help us to automatically derive
effect observations of the form
${\theta^{\mathcal{T}} : \mathcal{T}(\Id) \longrightarrow
\mathcal{T}(W)}$ (see \autoref{sec:EffObsExamples}), which
only slightly generalize the {\em DM4Free} construction.
In all other cases in \autoref{sec:overview}, we had to define effect observations by hand.
However, when the specification monad has a specific shape,
such as $\W^{\Pure}$, there is in fact a simpler way to define effect observations.
For instance, in \autoref{sec:observations} effect observations
$\theta^{\bot}, \theta^{\top} : \Exc \to \W^{\Pure}$ were used 
to specify the total and partial correctness of programs with
exceptions, by making a global choice of allowing or disallowing exceptions.
Here we observe that such hand-rolled effect observations can in fact be automatically
derived from $M$-algebras.\ch{Is there any intuitive explanation
  for what an $M$-algebra is?}


%
As shown by \citet{HylandLPP07}, there is a one-to-one
correspondence between monad morphisms $\M \to \Cont_R$ and $\M$-algebras $\M\,R
\to R$.
We can extend this to the ordered setting: for instance,  
effect observations $\theta : \M \to \W^{\Pure}$ correspond one-to-one to 
$\M$-algebras $\alpha : \M\,\prop \to \prop$ that are monotonic with respect
to the free lifting on $\M\,\prop$ of the implication order on $\prop$. 
Intuitively, $\alpha$ describes a global choice of how to assign a specification 
to computations in $M$ in a way that is compatible with  
$\retM$ and $\bindM$, e.g., such as disallowing all (or perhaps just some) exceptions.

Based on this correspondence, the effect observations $\theta^{\bot}$ and $\theta^{\top}$ arise from 
the $\Exc$-algebras $\alpha^\bot = \abs {} {\_} {\bot}$ and 
$\alpha^\top = \abs {} {\_} {\top}$.
Similarly, the effect observations for nondeterminism from
\autoref{sec:EffObsExamples} arise from the $\NDet$-algebras
$\alpha^{\forall}$ and $\alpha^{\exists}$, taking respectively the
conjunction and disjunction of a set of propositions in $\NDet{}(\prop)$,
as follows: $\theta^{\forall}(m) =
\abs{}{p}{\alpha^{\forall}\,(\NDet{}(p)\,m)}$ and $\theta^{\exists}(m) =
\abs{}{p}{\alpha^{\exists}\,(\NDet{}(p)\,m)}$.
Conversely, we can recover the $\NDet$-algebra $\alpha^{\forall}$ as
$\abs{}{m}{\theta^{\forall}_{\prop}(m)\,\mathrm{id}_{\prop}}$,
respectively $\alpha^{\exists}$ as
$\abs{}{m}{\theta^{\exists}_{\prop}(m)\,\mathrm{id}_{\prop}}$.

\da{State is now described separately below, and probability will 
probably need a treatment ala \cite{Voorneveld19}.}

Importantly, this correspondence is not limited to $\W^\Pure$,
but applies to continuation monads with any answer type. 
For instance, taking the answer type 
to be $S \to \prop$, we can recover the effect observation 
$\theta^{\St} : \St \rightarrow \W^\St$, where 
$\W^\St A \cong \mathrm{Mon}\Cont_{S \to \prop} A = 
(A \to (S \to \prop)) \to (S \to \prop)$, from the $\St$-algebra
$\alpha^{\St} = \abs {} {(f : S \to (S \to \prop) \times S)~(s : S)} {(\proj 1 \, (f\, s))\, (\proj 2 \, (f\, s))} 
  ~:~ \St (S \to \prop) \to S \to \prop$.



\section{Dijkstra monads from effect observations}
\label{sec:dijkstraMonads}
\newcommand{\DMon}{\mathcal{DM}\mathrm{on}}
\newcommand{\Mon}{\mathcal{M}\mathrm{on}}
\newcommand{\MonRel}{\mathcal{M}\mathrm{on}\mathcal{R}\mathrm{el}}
\newcommand{\fib}{\mathfrak{pre}}
\newcommand{\packFib}[0]{\smallint}


As illustrated in \autoref{sec:DijkstraMonadsExamples}, Dijkstra
monads can be obtained from 
effect observations $\theta : \M \to \W$ between 
computational and specification monads.
%
%
As we shall see this construction is generic and leads to a
categorical equivalence between Dijkstra monads and effect
observations.
In this section, we introduce more formally the notion of Dijkstra
monad using dependent type theory, seen as the internal language of a
comprehension category~\cite{Jacobs93}, and then build a category
$\DMon$ of Dijkstra monads.
In order to compare this notion of Dijkstra monads to effect
observations, we also introduce a category of monadic relations
$\MonRel$ and show that there is an adjunction
\begin{equation}
  \label{eq:intfibadj}
\packFib \;\;\dashv\;\; \fib \quad:\quad \MonRel\;\; \longrightarrow \;\;\DMon.
\end{equation}
Intuitively, an adjunction establishes a correspondence between
objects of two categories,
here $\MonRel$ and $\DMon$.
An adjunction always provides an equivalence of categories if we restrict our
attention to objects that are in one-to-one correspondence, those for which the
unit (resp. the counit) of the adjunction is an isomorphism.
When we restrict the adjunction above, we obtain an equivalence
between Dijkstra monads and effect observations.
For the sake of explanation, we proceed in two steps: first, we
consider Dijkstra monads and effect observations over specification
monads with a discrete order (\IE ordinary monads), describing
the above adjunction in this situation; later,
we extend this construction to general preorders, thus obtaining the actual adjunction we are
interested in.
We denote categories defined over non-discrete specification monads
with $\cdot^{\leq}$.

%

%
\newcommand{\TypeCat}{\mathcal{T}\!\!\mathit{ype}}
\newcommand{\OrdCat}{\mathcal{O}\hspace{-0.5mm}\mathit{rd}}
\newcommand{\Quo}{\mathcal{Q}}
\newcommand{\Discr}{\Delta}
\newcommand{\Forget}{\mathcal{U}}
\newcommand{\CoDiscr}{\mathrm{CoDiscr}}
\paragraph{Dependent type theory and comprehension categories}
We work in an extensional type theory with dependent products $(x:A) \to B$,
strong sums $(x:A) \times B$, an identity type $x =^{A} y$ for $x, y:A$ (where
the type $A$ is usually left implicit), 
a type of (proof-irrelevant) propositions
$\prop$, and quotients of equivalence relations.
This syntax is the internal language of a comprehension category~\cite{Jacobs93}
with enough structure and we will write $\TypeCat$ for any such category.
This interpretation of type theory allows us to call any object
$\Gamma \in \TypeCat$\et{pick something else than $\Gamma$? (suggests a type
  env)}\km{Yes objects of the base category are used to denote context usually,
  there is a bit of ambiguity between the context and the iterated sigma
  type/telescope built from it} a type.
\footnote{Under a mild condition that the
  category is \emph{democratic}~\cite{ClairambaultD14}.}

\paragraph{Dijkstra monads}
A \emph{Dijkstra monad} over a (specification) monad $\W$ is given by 
\begin{itemize}[leftmargin=*,nosep,label=$\vartriangleright$]
\item for each type $A$ and specification $w: \W\, A$, a type  $\mathcal{D}\;A\;w$
  of ``computations specified by $w$''
\item return and bind functions specified respectively by the return and bind
  of $\W$
  \begin{displaymath}
  \begin{array}{c}
    \retD : (x:A) \to \mathcal{D}\;A\;(\retW\;x)
    \\[1mm]
    \bindD : \mathcal{D}\;A\;w_{1} \to ((x:A) \to \mathcal{D}\;B\;w_{2}(x)) \to
    \mathcal{D}\;B\;(\bindW\; w_{1}\; w_{2})
  \end{array}
  \end{displaymath}
\item such that the following monadic equations about $\retD$ and $\bindD$ hold 
  \begin{displaymath}
  \begin{array}{c}
    \bindD\,m\,\retD = m
    \qquad
    \bindD\, (\retD\,x)\,f = f\,x
    \\[1mm]
    \bindD\,(\bindD\,m\, f)\,g = \bindD\,m\,(\abs{}x{\bindD\,(f\,x)\,g})
  \end{array}
  \end{displaymath}
  where $m:\mathcal{D}\,A\,w, x:A, f:(x:A) \to \mathcal{D}\,B\,(w'\,x), g :
  (y:B) \to \mathcal{D}\,C\,(w''\,y)$ for $A,B,C$ any types and $w:\W\,A, w':(x:A) \to \W\,B,
  w'': (y:B) \to \W\,C$.
  Note that the typing of these equations depends on the monadic equations
  for $\W$ and they would not be well-typed otherwise.
\end{itemize}
\bob{Examples here?}
%
%

%
In order to use multiple Dijkstra monads, that is multiple effects, in a
single program, we need a way to go from one to another, not only at the
 level of computations, but also at the level of specifications.
A morphism of Dijkstra monads from ${\mathcal{D}_{1}\,A\,(w_{1}:\W_{1}\,A)}$ to
${\mathcal{D}_{2}\,A\,(w_{2} : \W_{2}\,A)}$ provides exactly that: it is a pair $(\Theta^{\W},\Theta^{\mathcal{D}})$ of a monad morphism 
$\Theta^{\W} : \W_{1} \to \W_{2}$ mapping specifications of the source Dijkstra monad to
specifications of the target Dijkstra monad, and a family of maps
%
\[\Theta^{\mathcal{D}}_{A, w_{1}} \quad :\quad
  \mathcal{D}_{1}\,A\,w_{1} \quad \longrightarrow \quad \mathcal{D}_{2}\,A\,(\Theta^{\W}w_{1}) \]
 indexed by types $A$ and specifications $w_{1} : \W_{1}\,A$,
satisfying the following axioms
\begin{mathpar}
  \Theta^{\mathcal{D}}(\ret^{\mathcal{D}_{1}}\,x) =
  \ret^{\mathcal{D}_{2}}\,x,
  \and
  \Theta (\bind^{\mathcal{D}_{1}}\,m\, f) =
  \bind^{\mathcal{D}_{2}}\,(\Theta^{\mathcal{D}}\,m)\,(\Theta^{\mathcal{D}} \circ f).
\end{mathpar}
\km{Add an example of an actual dijkstra monad morphism}%
%
%
This gives a category $\DMon$ of Dijkstra monads and morphisms between them.

%
%

%
\paragraph{Monadic Relations} Given a monadic relation
$\mathcal{R} : \M \leftrightarrow \W$
(\autoref{def:monadicRelation}) between a computational monad
$\M$ and a specification monad $\W$, we construct a Dijkstra monad
$\fib ~\mathcal{R}$ on $\W$ as follows:
\begin{equation}
  \label{eq:fibDiscr}
  (\fib~\mathcal{R})~A~(w:\W\,A) \quad \eqdef \quad
  (m:\M\, A) \times m\,\mathcal{R}_{A}\,w
\end{equation}
\da{Add more whitespace between the arguments. A bit crammed at the moment. Also, 
do we really need the type info for $w$ on the left?}%
That is $(\fib~\mathcal{R})\,A\,w$ consists of those elements $m$
of $\M\,A$ that are related by $\mathcal{R}$ to the specification $w$.
When $\mathcal{R}$ is the graph of a monad morphism $\theta$ (or
equivalently, $\mathcal{R}$ is functional), $\fib (\mathcal{R}
: \M \leftrightarrow \W)$ maps an element $w : \W\,A$ to its preimage
$\theta^{{-1}}(w) = \aset{ m : \M\,A \alt \theta(m) = w}$.
Conversely, any Dijkstra monad $\mathcal{D}$ over $\W$ yields a
monad structure on
\[
\packFib \mathcal{D}\,A \eqdef (w:\W\,A) \times \mathcal{D}\,A\,w
\]
and the projection of the first component is a monad morphism $\pi_{1}
: \packFib \mathcal{D} \to \W$.
In order to explain the relation between these two operations $\fib$ and
$\packFib -$, we introduce the category $\MonRel$ of monadic relations.
An object of $\MonRel$ is a pair of monads $\M, \W$ together
with a monadic relation $\mathcal{R} : \M \leftrightarrow \W$
between them.
A morphism between $\mathcal{R}^{1} : \M_{1} \leftrightarrow
\W_{1}$ and $\mathcal{R}^{2} : \M_{2} \leftrightarrow \W_{2}$
is a pair $(\Theta^{\M}, \Theta^{\W})$ where $\Theta^{\M} : \M_1 \longrightarrow \M_2$ and $\Theta^{\W} : \W_1 \longrightarrow \W_2$ such that
\begin{equation}
\label{eq:MonRel}
\forall (m : M A)\, (w : W A).\, m\,\mathcal{R}^{1}_{A}\,w \implies \Theta^{\M}(m)\,\mathcal{R}^{2}_{A}\,\Theta^{\W}(w).
\end{equation}
The construction $\fib$ extends to a functor on $\MonRel$ by sending a
pair $(\Theta^{\W}, \Theta^{\M})$ to a pair $(\Theta^{\W},
\Theta^{\mathcal{D}})$, where $\Theta^{\mathcal{D}}_{A,w}$ is
the restriction of $\Theta^{\M}_{A}$ to the appropriate domain.
Conversely, $\packFib$ packs up a pair $(\Theta^{\W},
\Theta^{\mathcal{D}})$ as $(\Theta^{W}, \Theta^{\M})$,
where $\Theta^{\M}_{A}(w,m)= \Theta^{\mathcal{D}}_{A,w}(m)$.
Since $\Theta^{\M}$ maps the inverse image of $w$ to the inverse image
of $\Theta^{\W}(w)$, condition (\ref{eq:MonRel}) holds.
Moreover, this gives rise to a natural bijection
\[
  \MonRel(\packFib \mathcal{D}, \mathcal{R}) \qquad\cong\qquad \DMon(\mathcal{D},
  \fib\, \mathcal{R})
\]
that establishes the adjunction~(\ref{eq:intfibadj}).  We can restrict
(\ref{eq:intfibadj}) to an equivalence by considering only those objects
for which the unit (resp.~counit) of the adjunction is an
isomorphism.
Every Dijkstra monad $\mathcal{D}$ is isomorphic to its image
$\fib~(\packFib \mathcal{D})$, whereas a monadic relation
$\mathcal{R}$ is isomorphic to $\packFib (\fib~\mathcal{R})$ if and
only if it is functional, i.e., a monad morphism.
This way we obtain an equivalence of categories between $\DMon$ and
the category of effect observations on monads with discrete preorder.
%

%
\paragraph{The ordered setting}
Recall that in the examples of \autoref{sec:DijkstraMonadsExamples}, the Dijkstra monads
$\mathcal{D}\,A\,(w : \W\,A)$ we derived from effect observations 
${\theta : M \longrightarrow W}$ naturally made use of the order on 
$\W$ to compare programmer-provided specifications to type-inferred ones.
%
%
This order structure on $W$ can be naturally lifted to 
Dijkstra monads, by requiring $\mathcal{D}$ to be 
equipped with a \emph{weakening} structure
  \[\weaken \quad:\quad w_{1} \leq_{A} w_{2} \times \mathcal{D}\,A\,w_{1} \quad
    \longrightarrow \quad
    \mathcal{D}\,A\,w_{2}\]
  such that the following axioms hold (where we conflate the propositions
  $w_{1}{\leq} w_{2}$ and their 
  proofs) 
  \begin{mathpar}
    \weaken \langle w {\leq} w, m \rangle  = m,
    \and
    \weaken \langle w_{1} {\leq} w_{2} {\leq} w_{3}, m\rangle = \weaken \langle
    w_{2}{\leq} w_{3}, \weaken \langle w_{1} {\leq} w_{2}, m\rangle\rangle
  \end{mathpar}
  \begin{align*}
    \bindD\, (\weaken \langle w_m {\leq} w_m', m \rangle)\,
    (\abs{}{a}{\weaken&\langle w_f\,a{\leq} w_f'\,a, f\,a \rangle}) =\\
    &\weaken\langle \bindW\, w_m\,w_f{\leq}\bindW\, w_m'\,w_f', \bindD\, m\, f \rangle.
  \end{align*}
  \km{Provide the $\weaken$ for a specific Dijkstra monad ?}
\et{shouldn't it be  $w_{1} {\leq} w_{3}$ instead of $ w_{1} {\leq} w_{2} {\leq}
  w_{3}$? or, how is the 3-place witness defined?}\km{I was assuming that the
  reader would infer that it means ``the witness of $w_1 \leq w_3$
  obtained from the transitivity applied to $w_1 \leq w_2$ and $w_2 \leq w_3$''}\et{I'd suggest to write $w_1 \leq w_3$, and possibly explain (isn't that witness unique anyway?)}
%
Such pairs of an ordered monad $W$ and a Dijkstra monad $\mathcal{D}\,A\,(w : \W\,A)$ with a
weakening structure form a category $\DMon^{\leq}$, whose morphisms 
are pairs of a monotonic monad
morphism and a Dijkstra monad morphism preserving the weakening
structure.
%
%
%
Further, the definition of $\fib$ extends similarly straightforwardly to the 
ordered setting: given a monad morphism $\theta : \M \to \W$, we define 
%
%
\begin{equation}
  \label{eq:fibOrder}
  (\fib~\theta)~A~(w:\W\,A) \quad \eqdef \quad (m:\M\,A) \times {\theta(m) \leq_A w} 
\end{equation}
This definition coincides with (\ref{eq:fibDiscr}) when the order $\leq_A$ on $\W\,A$ is
discrete.
Moreover, we can equip $\fib~\theta$ with a weakening structure:
$\weaken \langle w_{1} \leq w_{2}, \langle m,\theta(m) \leq w_1\rangle\rangle \eqdef
  \langle m, \theta(m) \leq w_1 \leq w_2 \rangle$.
The same construction can be performed starting with an \emph{upward
  closed} monadic relation ${\mathcal{R} : \M \leftrightarrow \W}$, i.e., such
that $M$ has a discrete order and 
$
\forall m.\, \forall (w_{1} \leq^{\W}_{A} w_{2}).\,
 m\,\mathcal{R}_{A}\,w_{1}  \Rightarrow  m\, \mathcal{R}_{A}\,w_{2}.
$
Doing so, we obtain a functor $\fib : \MonRel^{\leq} \longrightarrow \DMon^{\leq}$ from the category of
upward-closed monadic relations to the category of ordered Dijkstra monads
with a weakening structure.
However, when trying to build a left adjoint $\packFib$ to $\fib$ exactly as before, 
there is a small mismatch with the expected construction on
practical examples.
Indeed, starting from a monad morphism ${\theta : \M \to
\W}$, $\packFib (\fib~\theta)$ reduces to $(\Sigma \M, \W,
\pi_{1})$ where $\Sigma\M\,A = (w:\W\,A) \times (m : \M\,A) \times\theta_{A}(m)\leq_{A} w$, 
which is unfortunately not isomorphic to $M$. The problem is that we 
get one copy of $m$ for each admissible specification $w : \W\,A$.
These copies, however, are non-essential since the weakening structure of 
$\fib~\theta$ identifies them. As such, to define $\packFib$, we need to further 
quotient them
\footnote{We conjecture that an alternative and
  more symmetric solution would be to equip our Dijkstra monads with an
  additional order, but this does not correspond to the examples we obtain in
  practice.}
, defining
\[\packFib{}\,\mathcal{D}\, A \eqdef{} \left ((w:\W\,A) \times \mathcal{D}\,A\,w \right ) / \sim \]
where $\sim$ is generated by $\langle {} w , c \rangle \!\sim\! \langle {} w',
\weaken{}\langle{} w \!\leq\! w', c \rangle \rangle$, 
%
giving us the desired adjunction $\packFib \dashv \fib$.
%

%

%
To summarize, we can construct Dijkstra monads with weakening out of effect
observations and the other way around.
Moreover, when starting from an effect observation $\theta : \M \to \W$, then
$\packFib (\fib~\theta)$ is equivalent to $\theta$.
This result shows that we do not lose anything when moving from effect
observations to Dijkstra monads, and that we can, in practice, use either the effect
observation or the Dijkstra monad presentation, picking the one that is most
appropriate for the task at hand.

\section{Algebraic effects and effect handlers for Dijkstra monads}
\label{sec:algebraic}

In \autoref{sec:jungle}, we noted that all our example computational 
monads come with corresponding canonical side-effect causing operations. 
This is an instance of a general approach to modeling computational effects algebraically 
using operations (specifying the sources of effects) and equations (specifying 
their behavior), as pioneered by \citet{PP:AOGE,PP:NCDM}. 
From the programmer's perspective, \emph{algebraic effects} naturally enable 
programming against an abstract interface of operations instead of a concrete 
implementation of a monad, with the accompanying notion of \emph{effect handlers} 
enabling one to modularly define different fit-for-purpose implementations of these 
abstract interfaces.


\subsection{Algebraic Effects for Dijkstra Monads}
\label{sec:algebraiceffects}


We begin by showing how effect observations naturally equip 
both the specification monad and the corresponding Dijkstra monad with
algebraic operations in the sense of \citet{PP:AOGE,PP:NCDM}. 
We observed several instances of this phenomenon for state, IO,
and nondeterminism in \autoref{sec:DijkstraMonadsExamples}, and we
can now explain it formally in terms of algebraic effects and effect observations.


\paragraph{Algebraic operations}
For any monad $M$, an \emph{algebraic operation} ${\algop : I \leadsto O}$
with input (parameter) type $I$ and output (arity) type $O$ is a family
$\algopM_A : I \times (O \to MA) \to MA$ that 
satisfies the following coherence law 
for all $i : I$, $m : O \to MA$, and  
$f : A \to MB$~\cite{PP:AOGE}:
\begin{equation}
  \label{eq:algopcoherence}
  \bindM\, (\algopM_A\, \langle i , m \rangle)\, f = \algopM_B\, \langle i , \abs{} o {\bindM\, (m\, o)\, f}\rangle
\end{equation}
For $\NDet$, the two operations are
$\choice : \one \leadsto \bool$ and $\fail : \one \leadsto \zero$.
For $\St$, the operations are $\stget : \one \leadsto S$ and
$\stput : S \leadsto \one$. Plotkin and Power also showed that such
algebraic operations are in one-to-one correspondence with
\emph{generic effects} $\geneffM : I \to M O$, which are 
often a more natural presentation for programming. For example, the
generic effect corresponding to the $\stput$ operation for $\St$ 
has type $S \to \St\,\one$. They are
interconvertible with algebraic operations as follows:
\begin{equation}
  \label{eq:geneffinterdefinable}
  \geneffM\, i \eqdef \algopM_O\, \langle i , \abs{} o \retM o \rangle
  \qquad
  \algopM_A\, \langle i,m \rangle \eqdef \bindM\, (\geneffM i)\, (\abs{} o m\, o)
\end{equation}
Plotkin and Power also show that signatures $\ii{Sig}$ of algebraic operations determine 
many computational monads (except continuations) 
once they are also equipped with suitable sets of equations $\ii{Eq}$. In the following we write 
$T_{(\ii{Sig},\ii{Eq})}$, abbreviated as $T$, for the monad determined by $(\ii{Sig},\ii{Eq})$.

\paragraph{Effect observations} In \autoref{sec:dijkstraMonads}, we
saw that Dijkstra monads are equivalent to effect
observations $\theta : M \to W$. When $M = T$, then since $\theta$ is a monad morphism,
it automatically transports any algebraic operations on the
computation monad $T$ to the (ordered) specification monad $W$:
\begin{equation}
  \label{eq:algopderivation}
  \algopW_A\, \langle i , w \rangle
  \eqdef
  \mu^{\W}_A\, (\theta_{W A}\, (\algopT_{WA}\, \langle i , \abs{} o {\retT (w\, o)} \rangle ))
\end{equation}
where $\mu^{\W} : W \circ W \to W$ is the \emph{multiplication} (or
\emph{join}) of $W$, defined as
$\mu^{\W}_A w \eqdef \bindW w\, (\abs{} {w'} {w'})$.

This derivation of algebraic operations is in fact a result of a more general 
phenomenon. Namely, given any monad morphism 
$\theta : M \to W$, we get a family of $M$-algebras on 
$W$, natural in $A$, by 
\begin{displaymath}
  \mu^{\W}_A \circ \theta_{W A}: M W A \longrightarrow W W A \longrightarrow W A
\end{displaymath}
Furthermore, the derived algebraic operations $\algopW_A$ (resp. the derived 
$M$-algebras on $W$) are monotonic with respect to the free lifting of 
the preorder $\leq^{\W\, A}$ on $WA$ to $T W A$ (resp. to $M W A$).

The derivation of operations on the specification monad from
operations on the computational monad, via the effect observation,
explains how we are able to systematically generate
(computationally natural) specifications
for operations in \autoref{sec:DijkstraMonadsExamples}. 
For instance, taking the effect observation 
$\theta^{\forall} : \NDet \to \W^{\Pure}$ for demonic nondeterminism, 
the induced operations we get on $\W^{\Pure}$ are
\begin{displaymath}
\begin{array}{l@{\qquad}l}
   \choice^{\W^{\Pure}}_A : \one \times (\bool \to \W^{\Pure} A) \to \W^{\Pure} A
   &
   \fail^{\W^{\Pure}}_A : \one \times (\zero \to \W^{\Pure} A) \to \W^{\Pure} A
   \\[1mm]
   \choice^{\W^{\Pure}}_A \langle () , w \rangle \eqdef \abs{} p {w\, \true\, p \wedge w\, \false\, p}
   &
   \fail^{\W^{\Pure}}_A \langle () , w \rangle \eqdef \abs{} p {\top}
\end{array}
\end{displaymath}

\da{currently I added the example of demonic nondeterminism, 
with the view that in the upcoming 6.3 something will be said 
about handlers for this effect. if that is not going to be the case, 
then perhaps we should either show exceptions or state as 
an example here}

\paragraph{Dijkstra monads}
Finally, we show that the Dijkstra monad $\mathcal{D} = \fib~\theta$ derived from a
given effect observation $\theta : T \to W$ in
(\ref{eq:fibOrder}) also supports algebraic operations, with their
computational structure given by the operations of $T$ and their
specificational structure given by the operations of $W$
derived in (\ref{eq:algopderivation}). This completes the process of lifting
operations from computational monads to Dijkstra monads that we
sketched in \autoref{sec:DijkstraMonadsExamples}.  In detail, we 
define an algebraic operation for $\mathcal{D}$ as
\begin{displaymath}
  \begin{array}{l}
    \algopD_A : (i : I) \to \big(c : (o : O) \to \mathcal{D}\, A\, (w\, o)\big)
    \to \mathcal{D}\, A\, (\algopW_A\, \langle i , w \rangle) \\[1mm]
    \algopD_A i\, c \eqdef 
    \big\langle ~ \algopT_A \langle i , \abs{} o {c\, o} \rangle ~ , ~ 
    \theta_A (\algopT_A \langle i , \abs{} o { c\, o} \rangle) \leq \algopW_A \langle i , w \rangle ~ \big\rangle
  \end{array}
\end{displaymath}

For instance, for $\NDd = \fib~\theta^{\forall}$, the induced operations 
have the following (expected) types:
\begin{displaymath}
\begin{array}{l}
   \choice^{\NDd}_A : (\_ : \one) \to \big(c : (b : \bool) \to \NDd\, A\, (w\, b)\big)
    \to \NDd\, A\, (\abs{} p {w\, \true\, p \wedge w\, \false\, p})
   \\[1mm]
   \fail^{\NDd}_A : (\_ : \one) \to \big(c : (z : \zero) \to \NDd\, A\, (w\, z)\big)
    \to \NDd\, A\, (\abs{} p {\top})
\end{array}
\end{displaymath}

As we have defined 
$\algopD$  in terms of algebraic operations for $T$ and $\W$, then it 
is easy to see that it also satisfies an appropriate variant of the 
algebraic operations coherence law (\ref{eq:algopcoherence}), namely 
\begin{displaymath}
    \bindD\, (\algopD_A\, \langle i , c \rangle)\, f = \algopD_B\, \langle i , \abs{} o {\bindD\, (c\, o)\, f}\rangle
\end{displaymath}

Finally, based on (\ref{eq:geneffinterdefinable}), we note that 
the generic effect corresponding to $\algop : I \leadsto O$ is given by 
\begin{displaymath}
   \geneffD\, i \eqdef \big\langle ~ \geneffT\, i ~,~ \theta_O\, (\geneffT\, i) \leq \geneffW\, i ~ \big\rangle 
     : I \to \mathcal{D}\, O\, (\geneffW\, i)
\end{displaymath}

\paragraph{Specifying operations in free monads} As we have seen above,
effect observations induce specifications for algebraic operations,
which in turn are used as the indices of the corresponding Dijkstra 
monad operations. Note that in the case of free monads, when 
$\ii{Eq} = \emptyset$, we have the freedom to assign arbitrary
specifications. Let $\mathrm{FreeM} \cong T_{(\ii{Sig},\emptyset)}$
be the free monad  over some signature \ii{Sig}:
\begin{displaymath}
  \mathrm{FreeM}~A = \mu X. A + \Sigma_{\algop_i : I_i \leadsto O_i \in \ii{Sig}}I_i \times (O_i \to X)
\end{displaymath}
To specify the operations, we assume that for each $\algop_i$, we have
a precondition $P_i : I_i \to \prop$ and a postcondition
$Q_i : I_i \times O_i \to \prop$. From these we can build
a $\mathrm{FreeM}$-algebra $h : \mathrm{FreeM}(\prop) \to \prop$ by
\begin{displaymath}
  \begin{array}{c}
    h(\inl\, \phi) = \phi \qquad\qquad
    h(\inr\, \langle \algop_i, \mathit{inp}, k \rangle) = P_i\, \mathit{inp} \land \forall \mathit{oup}.\, Q_i \langle \mathit{inp},\mathit{oup} \rangle \to k\, \mathit{oup}
  \end{array}
\end{displaymath}
Following \autoref{sec:monad-algebras}, we can derive  
an effect observation $\theta : \mathrm{FreeM} \rightarrow \W^{\Pure}$ from $h$, 
from which we can in turn derive a $\W^{\Pure}$-indexed Dijkstra monad
$\mathcal{D} = \fib\;\theta$ following \autoref{sec:dijkstraMonads}. 
The operations $\algop_i$ of the free
monad lift to generic effects in $\mathcal{D}$, with specifications
derived from the assumed $P_i$ and $Q_i$:
\begin{displaymath}
  \mathtt{gen}^{\mathcal{D}}_{\algop_i} : (\mathit{inp} : I_i) \to \mathcal{D}\;O_i\;(\lambda p.\, P_i\, \mathit{inp} \land \forall \mathit{oup}.\, Q_i\langle \mathit{inp},\mathit{oup}\rangle \to p\, \mathit{oup})
\end{displaymath}

As an example, consider the operation $\choice : \one \leadsto \bool$
introduced above but with the specification that always returns
$\true$. This is captured by the precondition $P_\choice\, \_ = \top$ and
postcondition $Q_\choice \langle \_ , b\rangle = (b = \true)$, which yields the following generic
effect (after simplifying its type):
\begin{displaymath}
  \mathtt{gen}^{\mathcal{D}}_{\choice} : \mathcal{D}\;\bool\;(\lambda p.\;p\, \true)
\end{displaymath}
In contrast to the demonic non-determinism specification of $\choice$ given
above for $\NDd$, this variant of $\choice$ always derives its weakest
precondition from the $\true$ case of the post condition.

In the next section, we will use the ability to assign arbitrary
specifications to operations in a free monad to define the proof
obligations required to verify effect handlers for those operations.

\subsection{Effect Handlers for Dijkstra Monads}

Of course, operations are only one side of  
algebraic effects: the other side concerns 
\emph{effect handlers} \cite{Plotkin13}. These are a generalization of exception handlers 
to arbitrary algebraic effects. They are defined by providing a concrete implementation 
for each (abstract) operation, such as $\stget$. 
Semantically, they denote user-defined $T$-algebras for the algebraic effect at hand.


In contrast to the general story for algebraic effects in
\autoref{sec:algebraiceffects}, our treatment of effect handlers for
Dijkstra monads is currently more ad hoc. We have two approaches,
which can roughly characterised in terms of how the
operations are assigned specifications. In the first approach, we do
not explicitly give a specification for each operation. Instead, the
specification is induced by the handler. This approach fits well with
the philosophy of Dijkstra monads with weakest precondition
specifications, i.e., automatically generating the most general specification
from the program. This approach works well for exceptions and allows us
to reconstruct the weakest precondition semantics for exceptions with
try/catch in the setting of Dijkstra monads
(\cite{Leino94,Sekerinski12}), and put the ad-hoc 
examples of \citepos{dm4free} {\em DM4Free} on a general footing.

Unfortunately, for \emph{resumable} operations
(i.e., everything except exceptions), the inevitable circularity
between the handler and the handled code leads to attempts to construct
inductive propositions that do not exist in Coq or
\fstar. To resolve this problem, we also demonstrate a second approach that
makes use of upfront specification of operations, as demonstrated at
the end of \autoref{sec:algebraiceffects}. This specification
of operations breaks the circularity, and allows handling of
operations that resume, such as $\choice$. However, this approach is
also not yet fully satisfactory: the operation clauses of the
handler must be verified \emph{extrinsically}, in contrast to the
usual methodology of Dijkstra monads.

\paragraph{Effect handling (1st approach)}
Following \citet{Plotkin13}, we define the \emph{handling} of $T$
(determined by some $(\ii{Sig},\ii{Eq})$) into a monad $M$
to be given by the following operation:
  \begin{displaymath}
    \begin{array}{@{}l}
    \handlewithi {T,M} : T\, A \to \big(h_{\algop_i} : I_i \times (O_i \to M B) \to M B\big)_{\algop_i : I_i \leadsto O_i \in Sig}
       \to (A \to M B) \to M B\\
    \handlewithi {T,M}\, (\retT_A\, a)\, h\, f \eqdef f\, a \\
    \handlewithi {T,M}\, (\algopT_A\, \langle i , t \rangle)\, h\, f \eqdef
      h_{\algop}\, \langle i , \abs{} o {\handlewithi {T,M}\, (t\, o)\, h\, f} \rangle
    \end{array}
  \end{displaymath}
where we leave the proof obligation that the operation 
clauses $h_{\algop}$ have to satisfy the equations in $Eq$ implicit. We refer 
the reader to \citet{Ahman18} for explicit treatment of such proof obligations.
  

As such, $h$ forms a $T$-algebra
$\alpha_{h} : T M B \to M B$, and
$\handlewithi {T,M}\, (-)\, h\, p\, f$ amounts to the induced
unique mediating $T$-algebra homomorphism
$ \alpha_{h} \circ T (f) : T A \to T M B \to M B $.

\paragraph{Specification monads} 
Based on the category theoretic view of effect handlers as user-defined $T$-algebras, 
we can define a notion of handling any monad $M$ into some other monad $M'$:
\begin{displaymath}
  \begin{array}{l}
    \handlewithi {M,M'} : M\, A \to (\alpha : M M' B \to M' B)  
      \to (A \to M' B) \to M' B
    \\
    \handlewithi {M,M'}\, m\, \alpha\, f \eqdef (\alpha \circ M(f)) \, m
  \end{array}
  \end{displaymath}
where we again leave the proof obligation ensuring that 
$\alpha$ is an $M$-algebra implicit.
%
%
Below we are specifically interested in $\handlewithi {M,M'}$ 
when $M$ and $M'$ are specification monads because, in contrast to
$T$, the structure of specification monads is not determined by 
$(\ii{Sig},\ii{Eq})$ alone.

\paragraph{Dijkstra monads}
Based on the smooth lifting of algebraic operations in \autoref{sec:algebraiceffects},
then when defining effect handling for the Dijkstra monad
$\mathcal{D} = \fib~\theta$ induced by some effect observation
$\theta : T \to W$ into some other Dijkstra monad $\mathcal{D}' = \fib~\theta'$
for $\theta' : M \to W'$, 
we would expect the computational
(resp. specificational) structure of handling to be given by that for
$T$ (resp. $W$).

However, simply giving an effect observation $\theta$ turns out to be
insufficient for handling $\mathcal{D}$ into $\mathcal{D}'$. Category
theoretically, the problem lies in the operation cases for $W$ giving
us a $T$-algebra $T W' B \to W' B$, but to use $\handlewithi {W,W'}$
(which we need to define the specification of handling)
we instead need a $W$-algebra $W W' B \to W' B$. To overcome this
difficulty, we introduce a more refined notion of effect observation,
relative to the specification monad $W'$ we are handling into.

\begin{definition}[Effect observation with effect handling]
An effect observation with effect handling for an ordered monad $W'$ is an effect 
observation $\theta : T \to W$ such that for any $T$-algebra \linebreak
$\alpha : T W' A \to W' A$, there is a choice of a $W$-algebra 
$\alpha_* : W W' A \to W' A$ that is (i) monotonic with respect 
to the orders of $W$ and $W'$, and (ii)  
which additionally satisfies the equation $\alpha_* \circ \theta_{W' A} = \alpha$.
\end{definition}

Intuitively, the condition (ii) expresses that 
$\alpha_*$ extends a $T$-algebra to a $W$-algebra 
in a way that is identity on the $T$-algebra structure, 
specifically on the algebraic operations corresponding to $\alpha$.

It is worth noting that needing to turn algebras 
$T W' A \to W' A$ into algebras 
$W W' A \to W' A$ is not simply a quirk  
due to working with Dijkstra monads, 
but the same exact need arises when giving a monadic semantics 
to a language with effect handlers 
using a monad different from $T_{(\ii{Sig},\ii{Eq})}$.

Using this refined notion of effect observation, we can now define 
handling for Dijkstra monads. Given an effect observation 
$\theta : T \to W$ with effect handling for $W'$ and another 
effect observation $\theta' : M \to W'$, we define the 
handling of $\mathcal{D} = \fib~\theta$ into  
$\mathcal{D}' = \fib~\theta'$ as the following operation
\begin{displaymath}
  \begin{array}{@{}l}
    \handlewithi {\mathcal{D},\mathcal{D}'}
    \begin{array}[t]{@{}c@{\hspace{0.3em}}l}
      :&\mathcal{D}\, A\, w_1 \\
      \to & \big(h^{W'}_{\algop_i} : I_i \times (O_i \to W' B) \to W' B\big)_{\algop_i : I_i \leadsto O_i \in Sig} \\
      \to & \big(h^{\mathcal{D}'}_{\algop_i} : \big( (i,c) : (i : I_i) \times \big((o : O_i) \to \mathcal{D}'\, B\, (w\, o)\big)\big) \to \mathcal{D}'\, B\, (h^W_{\algop} \langle i , w \rangle)\big)_{\algop_i \in Sig} \\
      \to & \big((a : A) \to \mathcal{D}'\, B\, (w_2\, a) \big) \\
      \to & \mathcal{D}'\, B\, (\handlewithi {W,W'}\, w_1\, (\alpha_{h^{W'}})_* \, w_2) 
    \end{array}
  \end{array}
\end{displaymath}
\begin{displaymath}
  \begin{array}{@{}l}
    \handlewithi {\mathcal{D},\mathcal{D}'}\, c_1\, h^{W'}\, h^{\mathcal{D}'}\, c_2
    \eqdef
    \big\langle~
    \begin{array}[t]{@{}l}
      \handlewithi {T,M}\, c_1\, h^{\mathcal{D}'}\, (\abs{} a {c_2\, a})\,  ~ , \\[0.5mm]
      \hspace{-2.75cm} \qquad \theta'_B\, \big(\handlewithi {T,M}\, c_1\, h^{\mathcal{D}'}\, (\abs{} a {c_2\, a})\big) \leq 
      \big(\handlewithi {W,W'}\, w_1\, (\alpha_{h^{W'}})_*\,  w_2\big) ~ \big\rangle
    \end{array}
  \end{array}
\end{displaymath}
where we again leave implicit the conditions ensuring that $h^{W'}$, $h^{\mathcal{D}'}$ are correct with respect to $\ii{Eq}$.
%
%

\paragraph{Exception handling}
One effect observation supporting handling is
${\theta^{\Exc} : \Exc \to \W^{\Exc}}$ from
\autoref{sec:observations}, where $\Exc$ is
determined by $(\{\throw\},\emptyset)$. To model
handling potentially exceptional computations into other similar
ones, as is often the case in languages with
exceptions but no effect system, we take
$\mathcal{D} = \mathcal{D}' = \EXC = \fib~\theta^{\Exc}$ and observe
that $\handlewithi {\EXC,\EXC}$ can be simplified to
\begin{align*}
  \mathtt{try\text{-}catch} : ~ & 
                                  \EXC\, A\, w_1 \to \big(h^{\W^{\Exc}}_{\throw} : E \to \W^{\Exc} B\big) \to  \big(h^{\EXC}_{\throw} : ( e : E) \to \EXC\, B\, (h^{\W^{\Exc}}_{\throw}\, e)\big) \\[-1mm]
  \to ~ & \big((a : A) \to \EXC\, B\, (w_2\, a) \big) \to \EXC\, B\, \big(\abs {} {p\, q} {w_1\, (\abs {} x {w_2\, x\, p\, q})\, (\abs {} e {h^{\W^{\Exc}}_{\throw}\, e\, p\, q} )} \big) 
\end{align*}
in part, by defining the extension
$\alpha_* : \W^{\Exc} \W^{\Exc} B \to \W^{\Exc} B$ of 
$\alpha : \Exc\, \W^{\Exc} B \to \W^{\Exc} B$ as
  \begin{equation}
  \label{eq:algebraextension}
    \alpha_*\, w
    \eqdef 
    \abs {} {p\, q} {\!w\, \big(\abs {} {{w'}\!} {\!\alpha\, (\inl\, {w'})\, p\, q}\big)\, \big(\abs {} {e} {\!\alpha\, (\inr\, e)\, p\, q}\big)} \\
    = 
    \abs {} {p\, q} {\!w\, \big(\abs {} {{w'}\!} {\!{w'}\, p\, q}\big)\, \big(\abs {} {e} {\!\alpha\, (\inr\, e)\, p\, q}\big)}
  \end{equation}
where the second equality holds because $\alpha$ is an 
$\Exc$-algebra and thus $\alpha\, (\ret^{\Exc}\, v) = \alpha\, (\inl\, v) = v$. 


On inspection, it turns out that $\mathtt{try\text{-}catch}$  
corresponds exactly to \citepos{Leino94} and \citepos{Sekerinski12} 
weakest exceptional preconditions for exception handlers.
Furthermore, with $\mathtt{try\text{-}catch}$ we can also put 
\citepos{dm4free} hand-rolled {\em DM4Free} 
exception handlers 
to a common footing. For 
example, we can define their integer division example as 
%
\vspace{-0.17cm}
\begin{lstlisting}
let div_wp (i j:int) = fun p q -> (forall x . j <> 0 /\ x = i / j ==> p x) /\ (forall e . j = 0 ==> q e)  
let div (i j:int) : EXC int (div_wp i j) = if j = 0 then raise div_by_zero_exn else i / j
let try_div (i j:int) : EXC int (fun p q -> forall x . p x) = try_catch (div i j) (fun _ p q -> p 0) (fun _ -> 0) (fun x -> x)
\end{lstlisting}
\vspace{-0.18cm}
where the specification of \ls{try_div} says that it never 
throws an exception, even not \ls{div_by_zero_exn}.


Of course $\mathtt{try\text{-}catch}$ is not the only way to handle exceptions. 
Another common use case is to handle a computation 
in $T A$ into one in $\Id (A + E)$. While this is 
trivial semantically, in a programming language 
where elements of $T$ are considered abstract, it allows 
one to get their hands on the values returned and 
exceptions thrown, analogously to \citepos{dm4free} use of monadic 
reification in {\em DM4Free}. To capture this, we take 
$\mathcal{D} = \EXC= \fib~\theta^{\Exc}$ 
and 
$\mathcal{D}' = \PURE = \fib~\theta^{\Pure}$, 
and define  
\begin{displaymath}
  \begin{array}{@{}l}
    \mathtt{reify} : ~ 
    \EXC\, A\, w \to \PURE\, (A + E)\, \big(\abs {} p {w\, (\abs {} x {p\, (\inl\, x)})\, (\abs {} e {p\, (\inr\, e)})} \big) 
    \\
    \mathtt{reify}\, c \eqdef 
    \handlewithi {\EXC,\PURE}\, c\, \big(\abs {} {e} {\!\ret^{\W^{\Pure}}\, (\inr\, e)}\big)  
    \big(\abs {} {e} {\!\ret^{\PURE}\, (\inr\, e)}\big)
    \big(\abs {} {x} {\!\ret^{\PURE}\, (\inl\, x)}\big)
  \end{array}
\end{displaymath}
%

\ch{How does this relate to the Ynot way of handling exceptions~\cite{ynot-icfp08},
  which is apparently inspired by \citet{BentonKennedy2001Exceptional}?}

\paragraph{Other (non-)examples}
Unfortunately, effect observations discussed in this paper other than
exceptions do not support effect handling. Specifically, we 
are unable to define the $\alpha_*$ operation for these 
effect observations, because it corresponds to attempting to construct
the specification of the handled computation knowing nothing of the
intended specification of the operations.


For IO, we actually know of another specification monad for which 
$\alpha_*$ can be defined, namely, the categorical coproduct of the IO 
and continuation monads \cite{HylandLPP07}, given by 
\[
   \W^{\IO} A \eqdef (A \to \prop) \times ((I \to \W^{\IO} A) \to \prop) \times (O \times \W^{\IO} A \to \prop) \to \prop
\]
Note that compared to the specification monads for IO from 
\autoref{sec:DijkstraMonadsExamples}, the postcondition(s) of $\W^{\IO}$ 
have a tree-like structure that enables one to recover enough information 
to (recursively) define $\alpha_*$.

There are however two major problems with using $\W^{\IO}$ as a 
specification monad. First, $\W^{\IO}$ is not 
well-defined in many categories of interest, such as $\mathsf{Set}$ \cite{HylandLPP07}. 
Second, defining $\W^{\IO}$ type theoretically requires non 
strictly positive inductive types, which leads to inconsistency 
in frameworks with impredicative universes such as Coq and
\fstar\cite{CoquandP88}.

\paragraph{Effect handling for upfront specified operations (2nd approach)}
We now describe an alternative approach to effect handling that avoids the 
above problems by making use of the upfront 
specified operations discussed at the end of
\autoref{sec:algebraiceffects}.
%
%
For simplicity, we assume that we are
handling into a pure computation of type $B$ with a postcondition
${R : B \to \prop}$.  We also assume that the computation to be handled
performs operations $\algop_i$ with the specifications $(P_i, Q_i)$ as
given above, yielding values of type $A$ satisfying some
postcondition $Q$, i.e., it has the type
$\mathcal{D}\;A\;(\lambda p.\;\forall a.\,Q\,a \to p\,a)$. 

The return
clause of the handler then gets to assume that $Q$ holds for its input
but must ensure that $R$ holds of its output. The operation cases of
the handler are more complex. We must first write each operation
clause without specification (i.e., as a function
$I_i \to (O_i \to B) \to B$), and then separately prove that, assuming
that the resumption argument is verified, then the final result is
verified. Note that we must separately program and verify the handler
clauses, contrary to the general methodology for programming with
Dijkstra monads. This is due to the higher-order nature of the
resumption argument.
Putting all this together, we get the following handling construct:
\begin{displaymath}
  \begin{array}{l@{}c@{\hspace{0.2em}}l}
    \mathtt{handle}
    &:  &\mathcal{D}\;A\;(\lambda p.\;\forall a.\,Q\, a \to p\,a) \\
    &\to&((a : A) \to Q\, a \to (b\mathord:B) \times R\, b) \\
    &\to&(h_i : I_i \to (O_i \to B) \to B)_{\algop_i} \\
    &\to&(\forall \mathit{inp}\,k.\, (\forall \mathit{oup}.\, Q_i \langle \mathit{inp}, \mathit{oup} \rangle \to R\,(k\,\mathit{oup})) \to P_i\, \mathit{inp} \to R\,(h_i\,\mathit{inp}\;k))_{\algop_i} \\
    &\to&(b\mathord:B) \times R\, b
  \end{array}
\end{displaymath}

\da{Should some of the $\to$s above be $\Rightarrow$s? 
Same question for the last bit of \autoref{sec:algebraiceffects} as well.}

For the ``always $\true$'' specification of $\choice$, we can write a
handler for it as $\lambda \_\,k.\,k\,\true$, which
yields the trivial proof obligation
${\forall k.\, (\forall b.\, b = \true \to R\,(k\,b)) \to R\,(k\,\true)}$. 
Note that this obligation would not hold if the
handler had relied upon invoking the resumption $k$ with $\false$. 

\newcommand\GenRec{\mathrm{GenRec}}
\newcommand\GenREC{\mathrm{GenREC}}
\newcommand{\call}[0]{\texttt{call}}

We used a variant of this second approach to verify programs with
general recursion in Coq, reconstructing from first principles
\fstar{}'s primitive support for total correctness, as well as its
semantic termination checking~\cite{mumon}.
Following \citet{McBride15}, we can describe a recursive function ${f :
(a:A) \to B\,a}$ by its body ${f_0 : (a:A) \to \GenRec\, (B\, a)}$, where
$\GenRec$ is the free monad on a single operation ${\call : (a:A) \leadsto B\,a}$
and the recursive calls to $f$ are replaced by uses of $\call$.
Given a well-founded order \ls{<<} on $A$, we ask that all
arguments to $\call$ are lower than the \emph{top-level} argument.
More precisely, given an invariant ${inv : (a:A) \to \W^\Pure\,(B\,a)}$ for $f$, we
define a family of effect observations ${\theta_a : \GenRec \to \W^\Pure}$ as described above, i.e., 
such that $\theta_a(\call\,a')$ strengthens $inv\,a'$ with the precondition $a'
\prec a$.
From these $\theta_a$s, we obtain a Dijkstra monad $\GenREC$
together with a handling construct
${\texttt{fix} : ((a:A) \to \GenREC\,(B\,a)\,(inv\,a)) \to (a:A) \to
  \PURE\, (B\,a)\, (inv\,a)}$. We have used this treatment of general 
  recursion to define and verify a simple Fibonacci example.

Compared to the ``specify at handling time'' approach above,
this ``specify upfront'' approach to effect handlers has the advantage that
it works for algebraic effects that involve resumptions. However,
there remain several unresolved questions with this approach,
including handling stateful computations and whether or not it is
possible to program and verify the handler clauses simultaneously to
be more in keeping with the general methodology of Dijkstra monads.

\section{Implementation and Formalization in \fstar{} and Coq}
\label{sec:implementation}

\paragraph{Dijkstra monads in \fstar{}}
We have extended the effect definition mechanism of \fstar{} to support
our more general approach to Dijkstra monads, in addition to the
previous {\em DM4Free} one.
\fstar{} users can now also define Dijkstra monads by providing both
a computational and a specification monad, along with an effect
observation or monadic relation between them, which provides more freedom
in the choice of specifications.
The SM language is not yet implemented in \fstar.
Nevertheless, this extension enables the verification of the examples of
\autoref{sec:DijkstraMonadsExamples}, for which effects such
as nondeterminism and IO were previously out of reach.
Once a Dijkstra monad is defined, the \fstar{} type-checker computes
weakest preconditions exactly as before and uses an SMT solver to
discharge them.
While internally \fstar{} only uses weakest preconditions as
specification monads, it is customary for users to write Hoare-style
pre- and postconditions, for which \fstar{} leverages the adjunction
from \autoref{sec:PrePostWPs}.

\paragraph{Dijkstra monads in Coq}
We have also embedded Dijkstra monads in Coq,
showing that the concept is applicable in languages beyond
\fstar. As with the \fstar implementation, programmers can supply their own
computational and specification monads, with an effect observation or monadic relation
between them. We implemented the base specification monads of
\autoref{sec:basic-spec-monads} and the construction of effect
observations from algebras of monads from 
\autoref{sec:monad-algebras},
thus providing a convenient way to build a specification monad and effect
observation at the same time.
The Coq development also constructs Dijkstra monads from effect
observations and proves their laws hold. Therefore, the
examples from \autoref{sec:DijkstraMonadsExamples} are verified in Coq
``all the way down''.  Verification in Coq follows the general pattern
of (a) writing the specification; (b) writing the program in monadic
style; and then (c) proving the resulting verification conditions using
tactic proofs. The Dijkstra monad setup automatically takes care of
the derivation of the weakest precondition transformer for the program.
\ch{Took this from the rebuttal, but the last sentence is not clear enough.}


\paragraph{Formalization of \SM{} in Coq}
We have formalized the \SM{} language of \autoref{sec:monad-transformers}
in Coq, taking Gallina as the base language $\bL$ and providing an
implementation of the denotation of SM terms and logical relation.
%
\SM{} is implemented using higher-order abstract syntax (HOAS) for the
$\abs{}{x}{t}$ binders and De Bruijn indices for the $\abs{\diamond}{x}{t}$ ones.
%
%
We build the functional version of the logical relation for a
covariant type $C$, but omit the linear type system.
Instead, the Coq version of \autoref{thm:MonadTransf} assumes a
semantic hypothesis requiring that the denotation of bind is
homomorphic, and using which it then derives the full monad
transformer (including all the laws).
\ifappendix
\autoref{sec:linearTyping} gives a paper proof that our syntactic
linearity condition entails this semantic hypothesis.
\else
A paper proof that our syntactic linearity condition entails the semantic
hypothesis can be found in the online appendix.
\fi
%
%
%

\ch{We should also explain that we can generate monad transformers (or
  pseudo) for all the example SM monads from Section 4.2. And we
  should explain to what extent these generated monad transformers are
  used in the verification examples.}

\section{Related work}
\label{sec:related}

This work directly builds on prior work on Dijkstra monads
in \fstar{}~\cite{fstar-pldi13, mumon}, in particular the {\em
  DM4Free} approach~\cite{dm4free}, which we discussed in detail in
\autoref{sec:intro} and \autoref{sec:predicateTransformers}.
Our generic framework has important advantages:
(1) it removes the previous restrictions on the computational monad;
(2) it gives much more flexibility in choosing the specification monad
and effect observation;
(3) it builds upon a generic dependent type theory, not on \fstar{} in particular.
\ch{Do we want to add the treatment of algebraic operations and
  handlers and to this list? In general, this list should be longer
  (see intro).}

\citet{Jacobs15} studies adjunctions between state transformers and
predicate transformers, obtaining 
a class of specification monads 
from the state monad transformer and an abstract notion of logical
structures.
He gives abstract conditions for the existence of such specification
monads and of effect observations.
\citet{HasuoGWP15} builds on the state-predicate adjunction of Jacobs
to provide algebra-based effect observations (in the style of 
\autoref{sec:monad-algebras})
for various computation and specification monads.
Our work takes inspiration from this, but provides a more concrete
account focused on covering the use of Dijkstra monads for
program verification.
In particular, we provide concrete recipes for building specification monads
useful for practical verification (\autoref{sec:metalanguage}).
Finally, we show that our Dijkstra monads{}
are equivalent to the monad morphisms built in these earlier works.



\citet{Katsumata14} uses graded monads to give semantics to
type-and-effect systems, introduces effect observations as monad
morphisms, and constructs graded monads out of 
effect observations by restricting the specification monads to their value
at $\one$.
We extend his construction to Dijkstra monads, showing that they are
equivalent to effect observations, and unify Katsumata's two 
notions of algebraic operation.
A graded monad can intuitively be seen as a non-dependent version of a
Dijkstra monad (a monad-like structure indexed by a monoid rather than a monad) 
but providing a unifying formal account is not
completely straightforward.
%
The framework of \citet{kaposi2019signatures} is a promising candidate for such
a unifying account that might provide an abstract proof of the results of
\autoref{sec:dijkstraMonads}
(see \ifcamera the online appendix\else
\autoref{sec:displayedAlgebras}\fi); we leave a full investigation as future work.

\citet{Katsumata13} gives a semantic account of \citet{Lindley2005}'s
$\top\top$-lifting, a generic way of
lifting relations on values to relations on monadic computations,
parameterized by a basic notion of relatedness at a fixed type. Monad
morphisms $MA \to ((A \to \prop) \to \prop)$, as used to generate
Dijkstra monads in \autoref{sec:DijkstraMonadsExamples}, are also unary
relational liftings $(A \to \prop) \to (MA \to \prop)$, and could be
generated by $\top\top$-lifting. Further, binary relational liftings
could be used to generate monadic relations that yield Dijkstra monads
by the construction in \autoref{sec:dijkstraMonads}. In both cases,
what is specifiable about the underlying computation would be
controlled by the chosen basic notion of relatedness.

\citet{RauchGS16} provide a generic verification framework for
{\em first-order} monadic programs.
Their work is quite different from ours, even beyond the
restriction to first-order programs, since their specifications are
``innocent'' effectful programs, which can observe the computational
context (\EG state), but not change it.
This introduces a tight coupling between computations and
specifications, while we provide much greater flexibility
through effect observations.
In fact, we can embed their work into ours, since their notion of
weakest precondition gives rise to an effect observation.

Generic reasoning about computational monads dates back 
to \citepos{Moggi89} seminal work, who proposes an
embedding of his computational metalanguage into higher-order logic.
Pitts \& Moggi's evaluation logic~\cite{pitts1991evaluation, Moggi95}
later introduces modalities to reason about the result(s)
of computations, but not about the computational context.
\citet{PlotkinP08} propose a generic logic for algebraic effects that
encompasses Moggi's computational $\lambda$-calculus, evaluation logic,
and Hennesy-Milner logic, but does not extend to Hoare-style
reasoning for state. 

\citet{SimpsonV18} and 
\citet{MatacheS19} explore logics for algebraic effects
by specifying the effectful behaviour of algebraic operations using 
a collection of effect-specific modalities instead of 
equations. Their modalities are closely related to how we 
derive effect observations $\theta : \M \to \W^{\Pure}$ and 
thus program specifications from 
$M$-algebras on $\prop$ in \autoref{sec:monad-algebras}, 
as intuitively the conditions they impose on their modalities 
ensure that these can be collectively treated as an 
$M$-algebra on $\prop$. 
In recent work concurrent to ours, \citet{Voorneveld19} 
studies a logic based on quantitative modalities by considering truth 
objects richer than $\prop$, including $S \to \prop$ for stateful and 
$[0,1]$ for probabilistic computation. While the state case we already 
briefly discussed in the context of deriving effect observations 
in \autoref{sec:monad-algebras}, it could be 
interesting to see if these ideas can be used to enable Dijkstra 
monads to be also used for reasoning about probabilistic programs.


In another recent concurrent work, \citet{SwierstraB19} study the predicate
transformer semantics of monadic programs with exceptions, state,
non-determinism, and general recursion. Their predicate transformer
semantics appears closely related to our effect observations, and their
compositionality lemmas are similar to our monad morphism laws.
We believe that some of their examples of performing verification directly
using the effect observation (instead of our Dijkstra monads), could
be easily ported to our framework.
Their goal, however, is to start from a specification and incrementally write
a program that satisfies it, in the style of the refinement
calculus~\cite{Morgan94}.
It could be an interesting future work direction to build a unified framework for
both verification and refinement, putting together the ideas of both works.

\da{In many ways what they are doing with derivations of partial programs 
also looks a lot like internalising typechecking rules for computations in 
Dijkstra monads. Not sure we want to say anything about this.}

\da{Also, we probably should say something about their attempt at 
partiality / general recursion as well.}

\section{Conclusion and Future Work}
\label{sec:conclusion}

This work proposes a general semantic framework for verifying programs
with arbitrary monadic effects using Dijkstra monads obtained from
effect observations, which are monad morphisms from a computation to a
specification monad.
This loose coupling between the computation and the specification
monad provides great flexibility in choosing the effect observation
most suitable for the verification task at hand.
We show that our ideas are general by applying them to both Coq and
\fstar{}, and we believe that they could also be applied to other
dependently-typed 
languages.

In the future, we plan to apply our framework to further computational
effects, such as probability~\cite{Giry}. 
It would also be interesting to investigate richer
specification monads, for instance instrumenting $W^\St$ with
information about framing, in the style of separation
logic.
Another interesting direction is to extend Dijkstra monads and our
semantic framework to relational reasoning, in order to obtain
principled semi-automated verification techniques for properties of
multiple program executions (\EG noninterference) or of multiple
programs (\EG program equivalence).
As a first step, we plan to investigate switching from (ordered)
monads to (ordered) relative monads for our specifications, by making return and
bind work on pairs of values.

Finally, the SM language provides a general way to obtain
correct-by-construction monad transformers, which could be useful in
many other settings, especially within proof assistants.
%
%
Categorical intuitions also suggest potential 
extensions of \SM{}, \EG some form of refinement types.






\ifanon\else
\begin{acks}                            

  We thank Nikhil Swamy and the anonymous reviewers for their feedback.\ch{Anyone else?}
  This work was, in part, supported by the
  \ifcamera
  \grantsponsor{1}{ERC}{https://erc.europa.eu/}
  \else
  \grantsponsor{1}{European Research Council}{https://erc.europa.eu/}
  \fi
  under ERC Starting Grant SECOMP (\grantnum{1}{715753}).
  Guido Mart\'inez' work was done, in part, during an internship at Inria Paris
  funded by the Microsoft Research-Inria Joint Centre.
  This material is based upon work supported by the 
  \grantsponsor{2}{Air Force Office of Scientific Research}{https://www.wpafb.af.mil/afrl/afosr/} 
  under award number \grantnum{2}{FA9550-17-1-0326}.
\end{acks}
\fi

\ifappendix
\appendix

\section{Appendix}
\subsection{Dijkstra monads as displayed algebras, relation to graded monads}
\label{sec:displayedAlgebras}
\newcommand{\bC}[0]{\mathbb{C}}
\newcommand{\metaArr}[0]{\,\hat{\Rightarrow}\,}
\newcommand{\bU}[0]{\mathbb{U}}
\newcommand{\El}{\mathbb{E}\mathrm{l}}
\newcommand{\Dom}[0]{\mathrm{Dom}}

The framework developed by \citet{kaposi2019signatures} can be used to
capture the notion of Dijkstra monad in a more concise way: they can
be seen as display algebras of a signature $\Sigma^{mon}$. Concretely,
Kov{\'a}cs proposed (in private communication) the following signature 
to capture Dijkstra monads:
\newcommand{\Set}[0]{\mathrm{Set}}
\[
\begin{array}{lcl}
  \M & : & \Set \metaArr \bU,\\
  \ret & : & (A:\Set) \metaArr A \metaArr \El (\M\,A),\\
  (-)^{\dagger} & : & (A B : \Set) \metaArr (\Pi_{A} \M\,B) \Rightarrow \M\,A \Rightarrow \El(\M\,A),\\
  \texttt{bind-ret}&:&(A : \Set) \metaArr (m:\M\,A) \Rightarrow \Id~(\M\,A)~(\ret^{\dagger}\;m)~ m,\\
  \texttt{ret-bind}&:&(A\; B : \Set)(x:A) \metaArr (f : \Pi_{A}\M\,B)  \Rightarrow \Id~(\M\,B)~(f^{\dagger}\;(\ret\; x))~(f\;x),\\
  \texttt{bind-assoc}&:&(A\; B\; C : \Set)\metaArr (m:\M A) (f :\Pi_{A} \M\,B)(g : \Pi_{B}\M\,C) \\
                         & & \quad\Rightarrow \Id~(\M\,C)~(g^{\dagger}\;(f^{\dagger}\;m))~((\abs {} x {g^{\dagger}(f x)})^{\dagger}\;m)
  
\end{array}
\]
Here $\Pi$ is here the constructor for infinitary ($A$-indexed for any $\Set~A$)
products. Taking models of this signature in the CwF of sets and families gives
monads on $\Set$, and unary logical predicate gives the notion of
Dijkstra monad without weakening.

The general equivalence between morphisms of context and 

Taking models in the CwF of preordered
sets, monotonic functions and fibrations of preorders gives the notion of
ordered monad and the unary logical predicate should provide the notion of
Dijkstra monad with weakening.

\km{TODO:Rework this !}
Dijkstra monads and part of their equivalence to effect observations can be
alternatively presented as an instance of a general result about displayed
algebras and categories with families.
Indeed there is a signature $\Sigma^{mon}$ in the sense of \citet{kaposi2019signatures},
generating a category with families of $\Sigma^{mon}$-algebras with monads as
algebras and Dijkstra monads as displayed algebras.
This category with families verifies the hypothesis of Prop. 9 of
\cite{ClairambaultD14}; thus the category of Dijkstra monads over a fixed monad
$\W$ is equivalent to the slice category over $\W$, that is monad morphisms with target $W$.



\section{Typing rules of \SM{}}

\autoref{fig:DMTypingAppendix} presents the typing rules of \SM{}, many of which 
we omitted in the body of the paper.

\begin{figure}
  \begin{mathpar}
    \inferrule{\Gamma \vdL A}{\Gamma \vdD \bM A}
    \and
    \inferrule{\Gamma \vdD C_{i}\quad i=1,2}{\Gamma \vdD C_{1} \times C_{2}}
    \and
    \inferrule{\Gamma \vdL A \\ \Gamma, x:A \vdD C}{\Gamma \vdD (x:A) \to C}
    \and
    \inferrule{\Gamma \vdD C_{i}\quad i=1,2}{\Gamma \vdD C_{1} \to C_{2}}
    \\
    \inferrule{ }{A \vdD \ret~:~ A \to \bM A}
    \and
    \inferrule{ }{A,B \vdD \bind~:~\bM A \to (A \to \bM B) \to \bM B}
    \and
    \inferrule{ }{\Gamma_{1},x:C,\Gamma_{2} \vdash x : C}
    \\
    \inferrule{\Gamma, x:A \vdD t : C}{\Gamma \vdD \abs {} x t : (x:A) \to C}
    \and
    \inferrule
    {\Gamma \vdL u : A\\ \Gamma \vdD t : (x:A)\to C}
    {\Gamma \vdD \app t u : \subst C u x}
    \and
    \inferrule{\Gamma, x:C_{1} \vdD t : C_{2}}{\Gamma \vdD \abs \diamond x t : C_{1} \to C_{2}}
    \\
    \inferrule
    {\Gamma \vdD \ttwo : C_{1}\\ \Gamma \vdD \tone : C_{1} \to C_{2}}
    {\Gamma \vdD \app \tone \ttwo : C_{2}}
    \and
    \inferrule
    {\Gamma \vdD t_{i} : C_{i}}
    {\Gamma \vdD \pair \tone \ttwo : C_{1} \times C_{2}}
    \and
    \inferrule{\Gamma \vdD t : C_{1}\times C_{2}}{\Gamma \vdD \proj i\; t : C_{i}}
  \end{mathpar}
  \vspace{-1em}
  \caption{Typing rules for \SM{}}
  \label{fig:DMTypingAppendix}
\end{figure}

\section{Proof of monotonicity of the denotations from \SM{}}

We start by expliciting the denotation from \SM{} to $\bL$ on terms.
We write $\elab{\M}{t}^{\gamma} : \elab{\M}{C}$ for the denotation of the term $\Gamma \vdash t
: C$ with respect to a monad $\M$ and substitution $\gamma : \elab{\M}{\Gamma}$.
\begin{mathpar}
  \elab{\M}{\ret}^{\gamma} = \ret^{\M}
  \and
  \elab{\M}{\ret}^{\gamma} = \bind^{\M}
  \and
  \elab{\M}{\pair \tone \ttwo}^{\gamma} = \pair{\elab{\M}{\tone}^{\gamma}}{\elab{\M}{\ttwo}^{\gamma}}
  \and
  \elab{\M}{\proj{i}\,t}^{\gamma} = \proj{i} \elab{\M}{t}^{\gamma}
  \and
  \elab{\M}{x}^{\gamma} = \gamma(x)
  \and
  \elab{\M}{\abs{\diamond}{x^{C_{1}}}{t}} = \abs{}{x^{\elab{\M}{C_{1}}}}{\elab{\M}{t}^{\gamma[x:=x]}}
  \and
  \elab{\M}{\app \tone \ttwo}^{\gamma} = \app{\elab{\M}{\tone}^{\gamma}}{\elab{\M}{\ttwo}^{\gamma}}
  \and
  \elab{\M}{\abs{}{x^{A}}{t}}^{\gamma} = \abs{}{x^{A}}{\elab{\M}{t}^{\gamma[x:=x]}}
  \and
  \elab{\M}{\app t u}^{\gamma} = \app{\elab{\M}{t}^{\gamma}}u
\end{mathpar}

\begin{figure}
  \begin{mathpar}
   \inferrule{}{\Gamma \vdD \bind~(\ret\,x)~f \equiv f\,x} 
   \and
   \inferrule{}{\Gamma \vdD \bind~m~\ret \equiv m} 
   \and
   \inferrule{}{\Gamma \vdD \bind~m~(\abs{}{x}{\bind~(f\,x)~g}) \equiv \bind~(\bind~m~f)~g} 
   \and
   \inferrule{}{\Gamma \vdD\proj{i}~\pair\tone\ttwo \equiv t_{i}}
   \and
   \inferrule{}{\Gamma \vdD\pair{\proj{1}~t}{\proj{2}~t}\equiv t}
   \and
   \inferrule{}{\Gamma \vdD(\abs{}{x}{t})\,u \equiv t\{u/x\}}
   \and
   \inferrule{}{\Gamma \vdD\abs{}{x}{t\,x} \equiv t}
   \and
   \inferrule{}{\Gamma \vdD(\abs{\diamond}{x}{\tone})\,\ttwo \equiv \tone\{\ttwo/x\}}
   \and
   \inferrule{}{\Gamma \vdD\abs{\diamond}{x}{t\,x} \equiv t}
   \\
   + \text{reflexivity, symmetry, transitivity and congruence for all term constructors}
  \end{mathpar}
  \vspace{-1.5em}
  \caption{Equational theory for \SM{}}
  \label{fig:EqTheorySM}
\end{figure}

We prove the two missing lemmas in the Coq development to extend to the
case when the order on the monad $\M$ used for the denotation is not discrete.
From these lemmas, we obtain that applying a monad transformer defined via an
internal monad in \SM{} to a specification monad is still a specification monad
and that the $\lift$ are monotonic.

\begin{theorem}[Monotonicity of denotation]
  Let $\M$ be an ordered monad, $\Delta ;\Gamma \vdD t : C$ a term in \SM{},
  $\vdL \delta : \Delta$ a substitution for the $\bL$ context $\Delta$, $(\vdL
  \gamma_{i} : \elab{\M}{\Gamma})_{i=1,2}$ substitutions for the \SM{} context $\Gamma$ such that $\forall (x:C_{0}) \in \Gamma.\,
  \gamma_{1}(x) \leq^{C_{0}} \gamma_{2}(x)$. Then $\elab{\M}{t}^{\delta;\gamma_{1}}
  \leq^{C} \elab{\M}{t}^{\delta;\gamma_{2}}$.
\end{theorem}
\begin{proof}
  By induction on the typing derivation of $t$:
  \begin{description}
  \item[Case $t = \ret_{A}: A \to \bM A$,] by reflexivity
    \[\elab{\M}{\ret_{A}}^{\delta;\gamma_{1}} = \retM_{A}\leq^{A
      \to \bM A} \retM_{A} = \elab{\M}{\ret_{A}}^{\delta;\gamma_{2}}\]
\item[Case $t = \bind_{A,B}:(A \to \bM B) \to (\bM A \to \bM B)$,] by reflexivity, that holds
    because $\bind^{M}$ is monotonic
  \[\elab{\M}{\bind_{A,B}}^{\delta;\gamma_{1}} = \bindM_{A,B} \leq^{(A \to \bM
      B) \to (\bM A \to \bM B)} \bindM_{A,B} = \elab{\M}{\bind_{A,B}}^{\delta;\gamma_{2}}\]    
\item[Case $t = \pair \tone \ttwo : A \times B$,] by induction
  \begin{mathpar}
    {\elab\M\tone}^{\delta;\gamma_{1}} \leq^{A} {\elab\M\tone}^{\delta;\gamma_{2}} 
    \and
    {\elab\M\ttwo}^{\delta;\gamma_{1}} \leq^{B} {\elab\M\ttwo}^{\delta;\gamma_{2}} 
  \end{mathpar}
  so
  \[
    {\elab\M{\pair\tone\ttwo}}^{\delta;\gamma_{1}} =
    \pair{{\elab\M\tone}^{\delta;\gamma_{1}}}{{\elab\M\ttwo}^{\delta;\gamma_{1}}}
    \leq^{A\times B}
    \pair{{\elab\M\tone}^{\delta;\gamma_{2}}}{{\elab\M\ttwo}^{\delta;\gamma_{2}}}
    = {\elab\M{\pair\tone\ttwo}}^{\delta;\gamma_{2}}
  \]
\item[Case $t = \proj i t' : A_{i}$,] by induction and extensionality
  \[\pair{\proj 1 {\elab\M{t'}^{\delta;\gamma_{1}}}}{\proj 2
    {\elab\M{t'}^{\delta;\gamma_{1}}}}
  = \elab\M{t'}^{\delta;\gamma_{1}} \leq^{A_{1} \times A_{2}}
  \elab\M{t'}^{\delta;\gamma_{2}}
  =
  \pair{\proj 1 {\elab\M{t'}^{\delta;\gamma_{2}}}}{\proj 2
    {\elab\M{t'}^{\delta;\gamma_{2}}}}
\]
so \[\proj i {\elab\M{t'}^{\delta;\gamma_{1}}} \leq^{A_{i}} \proj i {\elab\M{t'}^{\delta;\gamma_{2}}}\]
\item[Case $t = \abs{}xt : (x:A) \to C$,] by induction for any $v:A$,
  \[\elab\M{t'}^{\delta[x:=v];\gamma_{1}} \leq^{C\{v/x\}}
    \elab\M{t'}^{\delta[x:=v];\gamma_{2}} \]
  we conclude by reduction since
  \[\elab\M{\abs{}x{t'}}^{\delta;\gamma_{1}}\,v =
    (\abs{}y{\elab\M{t'}}^{\delta[x:=y];\gamma_{1}})\,v =
    \elab\M{t'}^{\delta[x:=v];\gamma_{1}}
  \]
  
\item[Case $t = {t'}\, v:C\{v/x\}$ ,] by induction
  \[\forall v_{0}:A.\, \elab\M{t'}^{\delta;\gamma_{1}}\,v_{0} \leq^{C\{v_{0}/x\}} \elab\M{t'}^{\delta;\gamma_{2}}\,v_{0} \]
  so
  \[\elab\M{{t'}\,v}^{\delta;\gamma_{1}} = \elab\M{t'}^{\delta;\gamma_{1}}\,v
    \leq^{C\{v/x\}} \elab\M{t'}^{\delta;\gamma_{2}}\,v = \elab\M{{t'}\,v}^{\delta;\gamma_{2}}\]
 
\item[Case $t= \abs\diamond x t' : C_1 \to C_2$,] for any $m_{1} \leq^{C_{1}}
  m_{2}$, $\gamma_{1}[x:=m_{1}] \leq^{\Gamma,x:C_{1}} \gamma_{1}[x:=m_{2}]$ and by induction 
  \[\elab\M{t'}^{\delta;\gamma_{1}[x:=m_{1}]} \leq^{C_{2}}\elab\M{t'}^{\delta;\gamma_{2}[x:=m_{2}]}\]
  and we conclude since for $i=1,2$
  \[({\elab\M{\abs\diamond{x}{t'}}^{\delta;\gamma_{i}[x:=y]}})\,m_{i} =
    (\abs{}y{\elab\M{t'}^{\delta;\gamma_{i}[x:=y]}})\,m_{i} =
    \elab\M{t'}^{\delta;\gamma_{i}[x:=m_{i}]}
 \]
 
\item[Case $t=\tone\,\ttwo:C_{2}$,] by induction hypothesis applied to $\ttwo :C_{1}$,
  \[\elab\M\ttwo^{\delta;\gamma_{1}} \leq^{C_{1}}
    \elab\M\ttwo^{\delta;\gamma_{2}}\]
  so by induction hypothesis applied to $\tone : C_{1} \to C_{2}$
  \[\elab\M\tone^{\delta;\gamma_{1}}~\elab\M\ttwo^{\delta;\gamma_{1}} \leq^{C_{2}}
    \elab\M\tone^{\delta;\gamma_{2}}~\elab\M\ttwo^{\delta;\gamma_{2}}\]
\end{description}
\end{proof}

\begin{theorem}[Monotonicity of relational interpretation]
  Let $\Delta \vdD C\,type$, $\M_{1}, \M_{2}$ two ordered monads and
  $(\mathcal{R}_{A})_{A}$ a family of \emph{monotonic} relations
  $R_{A} : \M_{1}\,A \times \M_{2}\,A \to \prop$, that is $R_{A}$ is an ideal
  wrt the order on $\M_{1}\,A \times \M_{2}\,A$, then
  $\elabrel{\M_{1}}{\M_{2}}{R}{C}$ is monotonic.
\end{theorem}
\begin{proof}
  by induction on the derivation of $C$ :
  \begin{description}
  \item[Case $C = \bM\,A$,]  $\elabrel{\M_{1}}{\M_{2}}{R}{\bM\,A} = R_{A}$ is
    monotonic by assumption
    
  \item[Case $C = C_1 \times C_2$,] suppose
    $(m_{1}, m_{2})~\elabrel{\M_{1}}{\M_{2}}{R}{C_{1} \times C_{2}}~(n_{1}, n_{2})$,
    $(m_{1}, m_{2}) \leq^{C_{1}\times C_{2}} (m'_{1}, m'_{2})$,
    $(n_{1}, n_{2}) \leq^{C_{1}\times C_{2}} (n'_{1}, n'_{2})$ then by induction
    hypothesis $m'_{1} \elabrel{\M_{1}}{\M_{2}}{R}{C_{1}} n'_{1}$ and
    $m'_{2} \elabrel{\M_{1}}{\M_{2}}{R}{C_{2}} n'_{2}$ so
    $(m'_{1}, m'_{2}) \elabrel{\M_{1}}{\M_{2}}{R}{C_{1} \times C_{2}} (n'_{1}, n'_{2})$
    
  \item[Case $C = (x:A) \to C'$,]
    suppose $f~\elabrel{\M_{1}}{\M_{2}}{R}{(x:A) \to C'}~g$, $f \leq^{(x:A) \to
      C'} f'$ and $g \leq^{(x:A) \to C'} g'$ then for any $v:A$,  
    $(f\,v)~\elabrel{\M_{1}}{\M_{2}}{R}{C'\{v/x\}}~(g\,v)$,
    $f\,v\leq^{C'\{v/x\}}f'\,v$, $g\,v\leq^{C'\{v/x\}}g'\,v$ so by inductive
    hypothesis $(f'\,v)~\elabrel{\M_{1}}{\M_{2}}{R}{C'\{v/x\}}~(g'\,v)$,
    hence $f'~\elabrel{\M_{1}}{\M_{2}}{R}{(x:A) \to C'}~g'$
   
  \item[Case $C = C_{1} \to C_{2}$,]
    suppose  $f~\elabrel{\M_{1}}{\M_{2}}{R}{C_{1} \to C_{2}}~g$, $f \leq^{C_{1} \to
      C_{2}} f'$ and $g \leq^{C_{1} \to C_{2}} g'$, for any $m~\elabrel{\M_{1}}{\M_{2}}{R}{C_{1}}~n$,
    $(f\,m)~\elabrel{\M_{1}}{\M_{2}}{R}{C_{2}}~(g\,n)$, $m \leq^{C_{1}} m$
    and $n \leq^{C_{2}} n$ so $f\,m \leq^{C_{2}} f'\,m$ and
    $g\,n\leq^{C_{2}}g'\,n$, hence by induction hypothesis 
    $(f'\,m)~\elabrel{\M_{1}}{\M_{2}}{R}{C_{2}}~(g'\,n)$
  \end{description}
\end{proof}

\section{Linear type system for \SM{}}
\label{sec:linearTyping}
\newcommand{\vdLin}[0]{\vdash_{\mathrm{lin}}}
We present a type system for \SM{} where contexts $\Gamma \alt \Xi$ are equipped
with distinguished position $\Xi$ called a stoup.
The stoup can be either empty or containing one variable of a type $C$ from
\SM{}.
Linear types are a refinements of types from \SM{} given by the following grammar
\[
  L := C \alt C_{1} \multimap C_{2} \alt L_{1} \times L_{2} \alt (x:A) \to L
  \alt L_{1} \to L_{2} 
\]
where $A \in \Type_{\bL}$, $C,C_{1},C_{2} \in \Type_{\SM}$.
In particular the linear function space $C_{1} \multimap C_{2}$ should be
understood as a subtype of $C_{1} \to C_{2}$ whose denotation ought to be a set
of homomorphisms with respect to the algebra structures on the denotations of its
domain and codomain, thus cannot be nested.
A linear judgement is of the form $\Delta ; \Gamma \alt \Xi \vdLin t : L$ with
the invariant that if $\Xi$ is non-empty then $\Xi = x:C_1$ and $L = C_2$ for
\SM{} types $\vdD C_1$ and $\vdD C_2$. 

\begin{figure}
\begin{mathpar}
  \inferrule{ }{A\alt - \vdLin \ret~:~ A \to \bM A}
  \and
  \inferrule{ }{A,B \alt - \vdLin \bind~:~\bM A \multimap (A \to \bM B) \to \bM B}
  \and
  \inferrule{ }{\Gamma \alt x : C \vdash x : C}
  \and
  \inferrule{(x:C) \in \Gamma}{\Gamma \alt - \vdash x : C}
  \and
  \inferrule
  {\Gamma \alt \Xi \vdLin t_{i} : C_{i}}
  {\Gamma \alt \Xi \vdLin \pair \tone \ttwo : C_{1} \times C_{2}}
  \and
  \inferrule{\Gamma \alt \Xi \vdLin t : C_{1}\times C_{2}}
  {\Gamma \alt \Xi \vdLin \proj i\; t : C_{i}}
  \and
  \inferrule{\Gamma, x:A \alt \Xi \vdLin t : C}
  {\Gamma \alt \Xi \vdD \abs {} x t : (x:A) \to C}
  \and
  \inferrule
  {\Gamma \alt \Xi \vdL u : A\\ \Gamma \alt \Xi \vdLin t : (x:A)\to C}
  {\Gamma \vdLin \app t u : \subst C u x}
  \and
  \inferrule{\Gamma \alt x:C_{1} \vdLin t : C_{2}}
  {\Gamma \alt - \vdLin \abs \diamond x t : C_{1} \multimap C_{2}}
  \and
  \inferrule{\Gamma, x:C_{1} \alt \Xi \vdLin t : C_{2}}
  {\Gamma \alt \Xi \vdLin \abs \diamond x t : C_{1} \to C_{2}}
  \and
  \inferrule{\Gamma \alt - \vdLin t : C_{1} \multimap C_{2}}
  {\Gamma \alt - \vdLin t : C_{1} \to C_{2}}
  \and
  \inferrule
  {\Gamma \alt \Xi \vdLin \ttwo : C_{1}\\ \Gamma \alt - \vdD \tone :
    C_{1} \multimap C_{2}}
  {\Gamma \alt \Xi \vdLin \app \tone \ttwo : C_{2}}
  \and
  \inferrule
  {\Gamma \alt - \vdLin \ttwo : C_{1}\\ \Gamma \alt \Xi \vdD \tone :
    C_{1} \to C_{2}}
  {\Gamma \alt \Xi \vdLin \app \tone \ttwo : C_{2}}
\end{mathpar}
\vspace{-1em}
\caption{Typing rules for SM with linearity condition}
\label{fig:lintyp}
\end{figure}

The value of this linear type system is provided by the following theorem:
\begin{theorem}[linear terms are homomorphisms]\label{thm:lth}
  Let $\M$ be a monad, $\Gamma \alt - \vdLin t : C_{1} \multimap C_{2}$ a term
  in $\SM{}$ and $\gamma : \elab{\M}{\Gamma}$, then the following diagram commutes
\[
  \begin{tikzcd}
    \M\elab{\M}{C_{1}} \ar[r, "\alpha^{C_{1}}_{\M}"] \ar[d, "\M\elab{\M}{t}^{\gamma}"] & \elab{\M}{C_{1}} \ar[d, "\elab{\M}{t}^{\gamma}"]\\
    \M\elab{\M}{C_{2}} \ar[r, "\alpha^{C_{2}}_{\M}"'] & \elab{\M}{C_{2}}
  \end{tikzcd}
\]
\end{theorem}

Indeed, for an internal monad $X \vdD C$ in \SM{}, the linearity condition on
$\bind^{C}$ requires a derivation of
\[ A,B\alt - \vdLin \bind^{C} : C\{A/X\} \multimap (A \to C\{B/X\}) \to C\{B/X\}\]
from which we can derive that
\[A,B, f: A \to B \alt - \vdLin
  \abs{\diamond}{x}{\bind^{C}~x~(\abs{}{y}{\ret^{C}\,(f\,y)})} : C\{A/X\}
  \multimap C\{B/X\}\]
that in turn proves that the right square in the diagram below
commutes thanks to \autoref{thm:lth}:
\[
  \begin{tikzcd}
    \M\,A \ar[r, "\M(\ret^{C}_{A})"] \ar[d, "\M\,f"'] &\M\elab{\M}{C}\,A \ar[r,
    "\alpha^{C}_{\M,A}"] \ar[d, "\M\elab{\M}{C}\,f"] & \elab{\M}{C}\,A \ar[d, "\elab{\M}{C}\,f"]\\
    \M\,B \ar[r, "\M(\ret^{C}_{B})"'] &\M\elab{\M}{C}\,B \ar[r, "\alpha^{C}_{\M,B}"'] & \elab{\M}{C}\,A
  \end{tikzcd}
\]
Thus, under the assumption that $\bind^C$ has a linear typing derivation (a
syntactic object), we prove that its denotation is homomorphic with respect to
the relevant $M$-algebra structure.

In order to prove the \autoref{thm:lth}, we need to :
\begin{itemize}
\item provide an interpretation of the linear types
\item show that linear derivations yield a denotation in this interpretation
\item prove using a logical relation that the linear interpretation of a term
  is related to its monotonic one
\end{itemize}
\newcommand\lelab[2]{(\!|#2|\!)_{#1}}
The interpretation of linear types is quite straightforward:
\begin{mathpar}
  \lelab\M{C} = \elab\M{C}
  \and
  \lelab\M{C_{1} \multimap C_{2}} = \aset{ f : \elab\M{C_1} \to \elab\M{C_2}
    \alt f \circ \alpha^{C_1}_{\M} = \alpha^{C_2}_{\M} \circ \M\,f}
  \and
  \lelab\M{L_{1} \times L_{2}} = \lelab\M{L_1} \times \lelab\M{L_2}
  \and
  \lelab\M{(x:A) \to L} = (x:A) \to \lelab\M{L}
  \and
  \lelab\M{L_{1} \to L_{2}} = \lelab\M{L_1} \to \lelab\M{L_2}
\end{mathpar}

\begin{theorem}[Denotation of linear typings]
  Let $\M$ be a monad, $\Delta ;\Gamma \alt \Xi \vdLin t : L$ a linear
  typing derivation of a term in \SM{},
  $\vdL \delta : \Delta$ a substitution for the $\bL$ context $\Delta$,
  $\vdL \gamma : \lelab{\M}{\Gamma}$ a substitution for the \SM{}
  context $\Gamma$, and $\vdL \xi : \lelab\M{\Xi}$. Then there is a well defined
  denotation $\lelab{\M}{t}^{\delta;\gamma;\xi} : \lelab{\M}{L}$ and if $\Xi =
  (x:C_1)$ then $L = C_2$ and $\abs{}{x}{\lelab{\M}{t}^{\delta;\gamma;x}} :
  \elab{\M}{C_1} \to \elab{\M}{C_2}$ is an $M$-algebra homomorphism.
\end{theorem}
\begin{proof}
 By induction on the linear typing derivation (each case corresponding to one
 derivation rule in \autoref{fig:lintyp}):
 \begin{description}
 \item[Case $t = \ret$,] $\lelab{\M}{\ret}^{\delta;\gamma;-} = \ret^\M : A \to
   \M A = \lelab{\M}{A \to \bM A}$
 \item[Case $t = \bind$,] $\lelab{\M}{\bind}^{\delta;\gamma;-} = \bind^\M : \M A \to
   (A \to \M B) \to \M B$
   with \[\M A \to (A \to \M B) \to \M B = \lelab\M{\bM A} \to \lelab{\M}{(A \to \bM B) \to \bM B}\]
   and $\bind^\M$ a homomorphism between the respective $M$-algebra structures

 \item[Case $t = x$ is linear,] $\lelab{\M}{x}^{\delta;\gamma;\xi} = \xi$ and
   the identity is an $M$-algebra map
\item[Case $t = x$ is not linear,] $\lelab{\M}{x}^{\delta;\gamma;-} = \gamma(x)$
  
\item[Case $t = \pair \tone \ttwo$,] $\lelab\M{\pair \tone \ttwo}^{\delta;\gamma;\xi} =
  \pair{\lelab\M{\tone}^{\delta;\gamma;\xi}}{\lelab\M{\ttwo}^{\delta;\gamma;\xi}}$ 
  and $\abs{}{\xi}{\pair {\lelab\M{\tone}^{\delta;\gamma;\xi}} {\lelab\M{\ttwo}^{\delta;\gamma;\xi}}}$ is an $M$-algebra map if and only if
  both $\abs{}{\xi}{\lelab\M{\tone}^{\delta;\gamma;\xi}}$ and
    $\abs{}{\xi}{\lelab\M{\ttwo}^{\delta;\gamma;\xi}}$ are $M$-algebra maps
\item[Case $t = \proj{i} t'$,] $\lelab\M{\proj{i} t'}^{\delta;\gamma;\xi} = \proj{i}{\lelab\M{t'}^{\delta;\gamma;\xi}}$ 
  and $\abs{}{\xi}{\proj{i}{\lelab\M{t'}^{\delta;\gamma;\xi}}}$ is an $M$-algebra
  map  whenever $\abs{}{\xi}{\lelab\M{t'}^{\delta;\gamma;\xi}}$ is an $M$-algebra map
\item[Case $t = \abs{}{x}{t'}$,] $\lelab\M{\abs{}{x}{t'}}^{\delta;\gamma;\xi} = \abs{}{x}{\lelab\M{t'}^{\delta[x:=x];\gamma;\xi}}$ 
  and $\abs{}{\xi}{\abs{}{x}{\lelab\M{t'}^{\delta[x:=x];\gamma;\xi}}}$ is an
  $M$-algebra map if and only if for any $\vdL x : A$,
  $\abs{}{\xi}{\lelab\M{t'}^{\delta[x:=x];\gamma;\xi}}$ is an $M$-algebra map
\item[Case $t = t' \,v$,] $\lelab\M{t'\,v}^{\delta;\gamma;\xi} = \lelab\M{t'}^{\delta;\gamma;\xi}\,v\{\delta\}$ 
  and $\abs{}{\xi}{\lelab\M{t'}^{\delta;\gamma;\xi}\,v\{\delta\}}$ is an
  $M$-algebra map whenever $\abs{}{\xi}{\lelab\M{t'}^{\delta;\gamma;\xi}}$ is an $M$-algebra map

\item[Case $t = \abs{\diamond}{x}{t'} : C_1 \multimap C_2$,]
  $\lelab\M{\abs{\diamond}{x}{t'}}^{\delta;\gamma;-} =
  \abs{}{x}{\lelab\M{t'}^{\delta;\gamma;x}} : \lelab\M{C_1} \to \lelab\M{C_2}$
  and it is and $M$-algebra map by induction hypothesis  
\item[Case $t = \abs{\diamond}{x}{t'} : L_1 \to L_2$,]
  $\lelab\M{\abs{\diamond}{x}{t'}}^{\delta;\gamma;\xi} =
  \abs{}{x}{\lelab\M{t'}^{\delta;\gamma[x:=x];\xi}}$ and
  $\abs{}{\xi}{\abs{}{x}{\lelab\M{t'}^{\delta;\gamma[x:=x];\xi}}}$ is an
  $M$-algebra map if and only if for any $\vdD x : \lelab\M{L_1}$,
 $\abs{}{\xi}{\lelab\M{t'}^{\delta;\gamma[x:=x];\xi}}$ is an $M$-algebra 
\item[Case $t : C_1 \to C_2$ is obtained from $t : C_1 \multimap C_2$,] the
  denotation of the term is the same, we just forget that it is an homomorphism
\item[Case $t = \tone\, \ttwo$, $\tone : C_1 \multimap C_2$,]
  $\lelab\M{\tone\,\ttwo}^{\delta;\gamma;\xi} =
  \lelab\M{\tone}^{\delta;\gamma;-}\,\lelab\M{\ttwo}^{\delta;\gamma;\xi}$
  and
  $\abs{}{\xi}{\lelab\M{\tone}^{\delta;\gamma;-}\,\lelab\M{\ttwo}^{\delta;\gamma;\xi}}$
  is an $M$-algebra map whenever
  $\abs{}{\xi}{\lelab\M{\ttwo}^{\delta;\gamma;\xi}}$ is an $M$-algebra map since
  $\lelab\M{\tone}^{\delta;\gamma;-}$ is an $M$-algebra map
\item[Case $t = \tone\, \ttwo$, otherwise,]
  $\lelab\M{\tone\,\ttwo}^{\delta;\gamma;\xi} =
  \lelab\M{\tone}^{\delta;\gamma;\xi}\,\lelab\M{\ttwo}^{\delta;\gamma;-}$
  and
  $\abs{}{\xi}{\lelab\M{\tone}^{\delta;\gamma;\xi}\,\lelab\M{\ttwo}^{\delta;\gamma;-}}$
  is an $M$-algebra map whenever
  $\abs{}{\xi}{\lelab\M{\tone}^{\delta;\gamma;\xi}}$ is an $M$-algebra map
\end{description}
\end{proof}

\newcommand{\relab}[2]{\,\langle{}\!| #2 |\!\rangle_{#1}\,}
Given a linear type $L$, we can forget all the linear annotations, obtaining a
type $\vdD |L|$ in \SM{}. In the same fashion, given a derivation $\Delta ;
\Gamma \alt \Xi \vdLin t : L$, we can obtain a derivation $\Delta ; |\Gamma,
\Xi| \vdD t : |L|$. In order to relate $\lelab{\M}{t}$ and $\elab\M{t}$, we
introduce the following relation $\relab\M{L} \subseteq \lelab{\M}{L} \times \elab\M{|L|}$ :
\begin{mathpar}
 m \relab\M{C} m' \iff m = m' 
 \and
 f \relab\M{C_1 \multimap C_2} f' \iff f = f' 
 \and
 \pair{x_1}{x_2} \relab\M{L_1 \times L_2} \pair{x_1'}{x_2'} \iff x_1
 \relab\M{L_1} x_1' \wedge x_2 \relab\M{L_2} x_2'
 \and
 f \relab\M{(x:A) \to L} f' \iff \forall (x:A).\, f\,x \relab\M{L} f'\,x
 \and
 f \relab\M{L_1 \to L_2} f' \iff (\forall x\, x'.\, x \relab\M{L_1} x' \rightarrow f\,x \relab\M{L} f'\,x')
\end{mathpar}
We extend componentwise this relation to context, and a straightforward but
tedious induction shows that for any linear derivation $\Delta ; \Gamma \alt
\Xi \vdLin t : L$ and context $\vdL \delta : \Delta$, $\vdL \gamma :
\lelab\M{\Gamma \alt \Xi}$, $\gamma' : \elab\M{|\Gamma, \Xi|}$, if $\gamma
\relab{\M}{\Gamma\alt \Xi} \gamma'$ then $\lelab\M{t}^{\delta;\gamma} \relab\M{L} \elab\M{t}^{\delta;\gamma'}$
where the right hand side denotation is obtained from the \SM{} derivation 
$\Delta ; |\Gamma, \Xi| \vdD t : |L|$. In the particular case where $\Xi$ is
empty and all type in $\Gamma$ do not contain any linear annotation, we obtain \autoref{thm:lth}.


\fi


\bibliographystyle{abbrvnaturl}
\bibliography{fstar}

\end{document}